\documentclass[a4paper,english, cleveref, thm-restate]{lipics-v2021}

\newcommand{\old}[1]{{}}

\hypersetup{
hidelinks,
colorlinks=true,
citecolor=[rgb]{1 0 0},
linkcolor=[rgb]{1 0 1},
urlcolor=[rgb]{1 0 0}
}
\usepackage{microtype} %

\title{Minimum Link Fencing}

\author{Sujoy {Bhore}}{Department of Computer Science \& Engineering, Indian Institute of Technology Bombay, India. \and\url{https://sites.google.com/view/homepage-of-sujoy-bhore/home}}{sujoy.bhore@gmail.com}{https://orcid.org/0000-0003-0104-1659}{}

\author{Fabian {Klute}}
{Department of Information and Computing Sciences, Utrecht University, Netherlands
\and \url{https://www.uu.nl/staff/FMKlute}}
{f.m.klute@uu.nl}{https://orcid.org/0000-0002-7791-3604}
{Supported by the Netherlands Organisation for Scientific Research (NWO) under
project no. 612.001.651 and the Austrian Science Foundation (FWF) grant J4510.}

\author
{Maarten L\"{o}ffler}
{Department of Information and Computing Sciences, Utrecht University, Netherlands
\and \url{https://www.uu.nl/staff/MLoffler/Profile}}
{m.loffler@uu.nl}
{}{}

\author{Martin {N\"ollenburg}}{Algorithms and Complexity Group, TU Wien, Austria \and \url{https://www.ac.tuwien.ac.at/people/noellenburg/}}{noellenburg@ac.tuwien.ac.at}{https://orcid.org/0000-0003-0454-3937}{}

\author{Soeren {Terziadis}}{Algorithms and Complexity Group, TU Wien, Austria \and \url{https://www.ac.tuwien.ac.at/people/sterziadis/}}{sterziadis@ac.tuwien.ac.at}{https://orcid.org/0000-0001-5161-3841}{}

\author{Ana\"is {Villedieu}}{Algorithms and Complexity Group, TU Wien, Austria \and \url{https://www.ac.tuwien.ac.at/people/avilledieu/}}{avilledieu@ac.tuwien.ac.at}{https://orcid.org/0000-0001-6196-8347}{Supported by the Austrian Science Fund (FWF) under grant P31119.}

\authorrunning{S. Bhore, F. Klute, M. L\"offler, M. N\"ollenburg, S. Terziadis and A. Villedieu}

\Copyright{Sujoy Bhore, Fabian Klute, Maarten L\"offler, Martin N\"ollenburg, Soeren Terziadis and Ana\"is Villedieu} %



\EventEditors{Sang Won Bae and Heejin Park}
\EventNoEds{2}
\EventLongTitle{33rd International Symposium on Algorithms and Computation (ISAAC 2022)}
\EventShortTitle{ISAAC 2022}
\EventAcronym{ISAAC}
\EventYear{2022}
\EventDate{December 19--21, 2022}
\EventLocation{Seoul, Korea}
\EventLogo{}
\SeriesVolume{248}
\ArticleNo{60}

\acknowledgements{The authors would like to thank Thekla Hamm and Irene Parada for the valuable discussion in particular concerning the \XP{} results of Section~\ref{sec:XPalgo}.}%

\nolinenumbers

\usepackage[usenames,svgnames,dvipsnames]{xcolor}
\usepackage{amsmath}
\usepackage{amsfonts}
\usepackage{amssymb}
\usepackage{graphicx}
\usepackage{subcaption}
\usepackage{cite}
\usepackage{hyperref}
\usepackage{xspace}
\usepackage{algorithmic}
\usepackage{amsmath}
\usepackage{graphicx,wrapfig}
\usepackage{amssymb,amsthm,amsmath}
\usepackage{complexity}
\usepackage{apptools}
\usepackage{etoolbox}
\usepackage{xpatch}

\makeatletter
\xpatchcmd{\thmt@restatable}%
{\csname #2\@xa\endcsname\ifx\@nx#1\@nx\else[{#1}]\fi}%
{\IfAppendix{\csname #2\@xa\endcsname}{\csname #2\@xa\endcsname\ifx\@nx#1\@nx\else[{#1}]\fi}}%
{}{} %
\makeatother

\newtheorem{remarkenv}{Remark}

\newenvironment{sketch}{%
  \renewcommand{\proofname}{Proof sketch}\proof}{\endproof}

\newcommand{\MLF}{Minimum Link Fencing\xspace}

\newcommand{\BMLF}{Bounded Minimum Link Fencing\xspace}
\newcommand{\SMLF}{Simply Bounded Minimum Link Fencing\xspace}
\newcommand{\CMLF}{Convex Bounded Minimum Link Fencing\xspace}

\newcommand{\fabianchange}[1]{{\color{black} #1}}
\newcommand{\soerenchange}[1]{{\color{black} #1}}

\newcommand{\isaacrev}[1]{{#1}}
\newif\ifAppendixProofs
\AppendixProofstrue

\newif\ifArxiv
\Arxivtrue

\ifArxiv \hideLIPIcs \else 
\fi

\newcommand{\prooffromappendix}[1]{\ifAppendixProofs {#1} \fi}

\newcommand{\sublab}[1]{\textbf{\textsf{(#1)}}}

\newcommand{\congruent}{congruent}

\newcommand{\colorletter}{\ensuremath{\kappa}\xspace}

\begin{document}

\maketitle

\begin{abstract}
    We study a variant of the geometric multicut problem, where we are given a set $\mathcal{P}$ of colored and pairwise interior-disjoint polygons in the plane.
    The objective is to compute a set of simple closed polygon boundaries (\emph{fences}) that separate the polygons in such a way that any two polygons that are enclosed by the same fence have the same color, and the total number of links
    of all fences is minimized.
    We call this the \emph{minimum link fencing} (\textsf{MLF}) problem
    and consider the natural case of
    \emph{bounded minimum link fencing} (\textsf{BMLF}), where $\mathcal{P}$ contains a polygon $Q$ that is unbounded in all directions and can be seen as an outer polygon. 
    We show that \textsf{BMLF} is \textsf{NP}-hard in general and
    \isaacrev{that it is \XP-time solvable 
    when each fence contains at most two polygons and 
    the number of segments per fence is the parameter.}
    \isaacrev{Finally, we present an $O(n \log n)$-time algorithm
    for the case that} the convex hull of $\mathcal{P} \setminus \{Q\}$ does not intersect $Q$.
\end{abstract}

\section{Introduction}

\begin{figure}
\centering

\begin{minipage}[t]{0.3\textwidth}
        \centering
        \includegraphics[page=2, width=\linewidth]{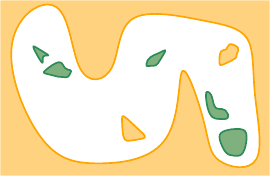}
    \end{minipage}
    \quad
    \begin{minipage}[t]{0.3\textwidth}
        \centering
        \includegraphics[page=3, width=\linewidth]{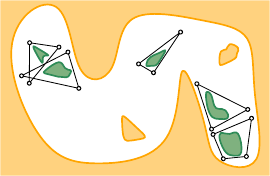}
    \end{minipage}
    \quad
    \begin{minipage}[t]{0.3\textwidth}
        \centering
        \includegraphics[page=4, width=\linewidth]{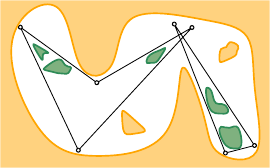}
    \end{minipage}

\caption
{
  Two sets of polygons in the plane (left) with different colors (green and yellow). The yellow set effectively acts as an outer polygon with holes. 
  Separating the two sets with, possibly intersecting, individual fences (middle) can lead to significantly more links in the fences (here 16) than grouping same-colored polygons (right), which achieves this with just seven links.
}
\label{fig:example}
\end{figure}

In the {\em geometric multicut} problem~\cite{aglr-gm-icalp19}, we are given $\colorletter$ disjoint sets of polygons in the plane, each with a different {\em color}, and are asked for a subdivision of the plane such that no cell of the subdivision contains multiple colors.
The goal is to minimize the total length of the subdivision edges.

A different kind of separation is achieved in the {\em polygon nesting} problem~\cite{DBLP:journals/iandc/YapABO89}, where for two polygons $P$ and $Q$ with $P\subset Q$ one asks for a  polygon $P'$ with the smallest number of links, s.t.\ $P\subset P'\subset Q$.
There exists a series of work that addressed the algorithmic complexity of nesting problems for various polygon families~\cite{DBLP:journals/iandc/YapABO89, ghosh1991computing, wang_finding_1991, das1991approximation,
mitchell1995separation}. 
See Section~\ref{relatedwork} for more detail.

In this paper, we consider a variant of geometric multicut inspired by polygon nesting, where we separate the sets from each other with a set of closed polygon \isaacrev{boundaries} called \emph{fences}, which enclose only polygons of one color and have the smallest possible number of links.
If one or more sets are not connected, %
we need to solve the combinatorial problem of choosing which polygons should be grouped in each fence.
\cref {fig:example} illustrates the problem.
Some variants of the fencing problem already become \textsf{NP}-hard for point objects with two colors, e.g., if we require the fence to be a single closed curve~\cite{DBLP:journals/prl/EadesR93}.

In this paper, we assume the input sets are collections of polygons, one color covers the plane minus a single polygonal hole (the outer polygon, a parallel to polygon nesting), and we will focus on the case $\colorletter=2$ of two colors.
We use $n$ to denote the total number of corners of the input polygons.
Even in this simple setting the problem turns out to be non-trivial.
If both sets are connected, then the problem is equivalent to finding a minimal {\em nested} polygon, which can be solved in $O(n \log n)$ time~\cite {AGGARWAL198998}.
If both sets are not connected we show this problem to be \textsf{NP}-hard in \cref{sec:nphard}.
Note the contrast to the geometric multicut problem, which is polynomially solvable for $\colorletter=2$~\cite{abrahamsen2018fast} but becomes \NP-hard when $\colorletter=3$~\cite{aglr-gm-icalp19}.
In \cref{sec:XPalgo} we show that, when restricting every fence to contain at most two polygons, the problem admits an \XP-algorithm when parameterized by the maximal number of segments per fence,
a result which holds for any $\colorletter$.
Finally, in \cref{sec:algorithm}, we show that the problem is polynomial-time solvable 
if the convex hull of the second color (the {\em inner polygons}) is contained in the outer polygon and the first color is connected.

\subsection {Problem Definition}
Throughout this paper we consider polygons in $\mathbb R^2$ without
self-intersections but potentially with holes.
Moreover, we consider a polygon as the boundary together with its interior, unless stated otherwise.
We consider the following problem.

\begin{definition}[\MLF (MLF)]\label{def:multicut_adapt}
We are given $n$ pairwise interior-disjoint polygons $\mathcal{P} = \{P_1,\ldots,P_\soerenchange{{|\mathcal{P}|}}\}$
in the plane, with a coloring function $f: \mathcal{P} \rightarrow \{1,\dots,\colorletter\}$, which assigns a color to every input polygon.
We write $\mathcal{P}_i = \{P \mid f(P) = i\}$.
We want to find a set of simple closed polygon boundaries $\mathcal{F} = \{{F_1}, \ldots, {F_m}\}$ 
such that the total number of links $|F|$ on the boundary of $F = \bigcup_{i=1}^k {F_i}$ is minimized
and if two polygons $P_a$ and $P_b$ are enclosed by the same fence or are both in $\mathbb{R}^2\setminus \bigcup_{i=1}^k \overline{F_i}$, where $\overline{F_i}$ is the polygon bounded by $F_i$, 
then $f(P_a)= f(P_b) $.
We  
call $F_i$ a fence and $\mathcal{F}$ a \emph{minimum link fencing} of $\mathcal{P}$.
\end{definition}

\isaacrev{Note the important difference in Definition~\ref{def:multicut_adapt} between $\mathcal{F}$, which is the set of all fences of a solution, and $F$, which is the union over all fences, i.e., one (\fabianchange{possibly} disconnected) polygon.
Thus $|\mathcal{F}|$ is the number of fences and $|F|$ is the number of all segments in these fences.}

\begin{figure}
\centering
\begin{subfigure}[t]{0.3\linewidth}
	\centering
	\includegraphics[page=8]{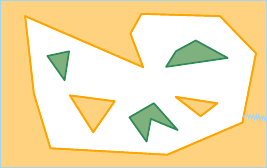}
	\subcaption{\textsf{BMLF}}
	\label{fig:setting_BMLF}
\end{subfigure}
\quad
\begin{subfigure}[t]{0.3\linewidth}
	\centering
	\includegraphics[page=5]{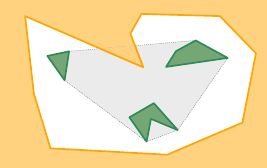}
	\subcaption{\textsf{SMLF}}
	\label{fig:setting_SMLF}
\end{subfigure}
\quad
\begin{subfigure}[t]{0.3\linewidth}
	\centering
	\includegraphics[page=6]{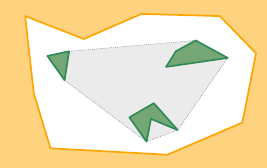}
	\subcaption{\textsf{CMLF}}
	\label{fig:setting_CMLF}
\end{subfigure}
\caption{Different problem inputs corresponding to (a) \textsf{BMLF}, (b) \textsf{SMLF} and (c) \textsf{CMLF}. In (b) and (c) the convex hull of all input polygons indicated in gray.}
\label{fig:settings}
\end{figure}

Throughout the paper we refer to $\mathbb{R}^2\setminus \bigcup_{i=1}^{|\mathcal{P}|} P_i$ as the \emph{free space} (between polygons). 
We refer to $\mathcal{P}$, which contains polygons of $\colorletter$ different colors, as $\colorletter$-colored and to the problem setting as the $\colorletter$-colored problem.
We consider several problem variations. 

If there exists a polygon $Q \in \mathcal{P}$ 
which is unbounded in every direction, i.e. $\mathbb{R}^2 \setminus Q$ is finite, this polygon $Q$ effectively acts as an outer boundary.
In this case we call the problem \BMLF (\textsf{BMLF}).
We denote the polygon $Q$ as the \emph{outer polygon}.
\isaacrev{\fabianchange{As a consequence, the size of the outer polygon automatically bounds the length of any link in a fence. Else, in general, one fence could conain a very long link}, while retaining small complexity \fabianchange{when counting the number of links only}.} 
Note that $Q$ can be emulated in an instance of \MLF, by adding a large rectangular polygon $P_c\setminus (\mathbb{R}^2 \setminus Q)$, i.e., a large rectangle, of which the area, which did not belong to $Q$ is cut out (light blue channel in Figure~\ref{fig:setting_BMLF}).
If $Q$ is the only polygon of its color $f(Q)$ we call this setting \SMLF (\textsf{SMLF}).
Moreover, if in an instance of \textsf{SMLF} we have $CH(\bigcup_{i=1}^\colorletter \mathcal P_i\setminus Q) \subset \mathbb{R}^2\setminus Q$,
i.e., the convex hull of all input polygons except $Q$ does not intersect $Q$, we speak of \CMLF (\textsf{CMLF}).
The differences are illustrated in \cref{fig:settings}.

\subsection{Related Work.}\label{relatedwork}
Despite the fact that the problem is natural and fundamental,
little previous work exists.
The problem of \emph{enclosing} a set of objects by a shortest system of
fences has recently been considered with a single set
$B_1$~\cite{abrahamsen2018fast}. The task is to ``enclose'' the
components of $B_1$ by a shortest system of fences.
This can be formulated as a special case of our problem with
$\colorletter=2$ colors: We add an additional
set $B_2$, far away from $B_1$
and large enough so that
it is never optimal to surround
$B_2$.
Thus, we have to enclose all components of $B_1$ and separate them from
the unbounded region. In this setting, there will be no nested fences.
Abrahamsen et al.~\cite{abrahamsen2018fast} gave an $O(n\;\text{polylog}\;n)$-time algorithm for inputs that consist of $n$ unit disks.

Some variations with additional constraints on the fence become
\NP-hard already for point objects with two colors. For example, if we require the fence to
be a single closed curve, it has been observed by
Eades and Rappaport~\cite{DBLP:journals/prl/EadesR93} already in 1993
that one can model the Euclidean Traveling Salesman Problem
of computing the shortest tour through a given set of sites
 by placing
two tiny objects of opposite color next to each site.
If we require the fence to be connected, the same construction will
lead to the Euclidean Steiner Tree Problem, which was shown to be \NP-hard by Garey et al.~in 1977~\cite{garey1977complexity}.

\smallskip\noindent\textbf{Polygon Nesting \& Separation.} Polygon nesting is considered to be a fundamental problem in computational geometry, and has been extensively studied since its inception. Aggarwal et al.~\cite{DBLP:journals/iandc/YapABO89} considered the problem of finding a polygon nested between two given convex polygons that has a minimal number of vertices. They gave an $O(n\log k)$ time algorithm for solving the problem, where $n$ is the total number of vertices of the given polygons, and $k$ is the number of vertices of a minimal nested polygon. Das~\cite{das1991approximation} considered a variant of \textsf{MLF} in his thesis, which restricts every fence to enclose exactly one polygon, and showed that the problem is \NP-hard.  
Given a polygon $Q$ of $m$ vertices inside another polygon $P$ of $n$ vertices,
Ghosh~\cite{ghosh1991computing} gave an $O((n+m)\log k)$ time algorithm for constructing a minimum nested convex polygon, where $k$ is the number of vertices of the output polygon, improving upon the $O((n+m)\log(n+m))$ time algorithm of Wang and Chan~\cite{WangC86}. 
However, on the other hand, given a family of disjoint polygons 
$P_1, P_2,\ldots, P_k$ in the plane, and an integer parameter $m$, it
is \textsf{NP}-complete to decide if the $P_i$’s can be pairwise separated by a polygonal family with at most $m$ edges. 
Mitchell and Suri~\cite{mitchell1995separation} presented efficient approximation algorithms for constructing separating families of near-optimal size.  

\medskip

\textit{Full proofs of statements marked by ($\star$) are found in the \ifArxiv Appendix\else full paper~\cite{arxiv}\fi.}

\section{Two-colored \textsf{BMLF} is \NP-hard}
\label{sec:nphard}
In this section we will call polygons of color 1 \emph{boundary polygons} and polygons of color 2 \emph{inner polygons}.
An instance of planar $3,4$-\SAT consists of a Boolean CNF-formula $\phi $ with a set of variables $ \mathcal{V} = \{v_1, \dots,  v_n\} $ and a set of clauses $ \mathcal{C} \subset 2^{\mathcal{V}} $, s.t.\ every clause is a disjunction of three literals and every variable occurs at most four times as a literal in a clause.
Additionally, we are given the embedded plane incidence graph $G_\phi = (\mathcal{V} \cup \mathcal{C}, E)$, where $E=\{vc \mid v \in \mathcal V, c \in \mathcal C, v \text{ occurs as a literal in } c\}$. %
\isaacrev{\fabianchange{
It is known that deciding if a $3,4$-\SAT-formula has a satisfying assignment is \NP-complete}~\cite{jansen1995minimum}.}

Given an instance of planar $3,4$-\SAT{} we create an instance of 2-colored \textsf{BMLF} $\mathcal{P}$, emulating the shape of
$G_\phi$ with one unbounded outer polygon $Q$ and multiple boundary polygons of the same color $f(Q) = 1$ (\cref{fig:schematized_construction}), s.t.\ $\phi$ is satisfiable if and only if there exists a minimum link fencing for $\mathcal{P}$ with at most a certain fixed number of total segments.

Note that each gadget is described as a basic construction of gray polygons, in which inner polygons are placed.
This is possible, because we will invert all gray polygons at the end of the reduction, s.t.\  the area of their union makes up exactly the actual free space of our entire construction, see \cref{fig:schematized_construction}.
Stating that fences are computed inside the gray polygons should be understood as fences being placed in the free space between polygons.
Throughout this reduction we distinguish fences based on the inner polygons they include.
We call two fences $F$ and $F'$ \emph{\congruent}, if and only if they enclose the same set of inner polygons.
We call two fencings $\mathcal{F}$ and $\mathcal{F}'$ \emph{\congruent} if there is a bijective mapping $f: \mathcal{F} \rightarrow \mathcal{F}'$, s.t., every $F\in \mathcal{F}$ is \congruent{} to $f(F) \in \mathcal{F}'$.

Let $\mathcal P$ be an instance of \textsf{BMLF} and
$S_1$, $S_2$, and $S_3$ disjoint connected subsets of $\mathbb R^2 \setminus \cup_{P\in\mathcal P} P$.
We call the ordered set $\mathcal S = \{S_1, S_2, S_3\}$ a \emph{non-collinear triple} 
if there \isaacrev{are} no three points $p_1\in S_1$, $p_2\in S_2$, and $p_3\in S_3$ 
such that the straight-line segment $s$ from $p_1$ to $p_3$ contains $p_2$
and $s$ \isaacrev{lies completely} inside $\mathbb R^2 \setminus \cup_{P\in\mathcal P} P$.
\isaacrev{T}he choice of $S_2$ only matters if there exists a straight-line segment in $\mathbb R^2 \setminus \cup_{P\in\mathcal P} P$
connecting points in $S_1$ and $S_3$.
\isaacrev{Therefore} we can often omit $S_2$ from the description of the triple or 
assume it as arbitrarily chosen.
We call $S_2$ the \emph{bend-set} of $\mathcal S$.
Let $\mathcal F$ be a fencing of $\mathcal P$, and $S_1$, $S_2$, and $S_3$ a non-collinear triple.
We say a fence $F \in \mathcal F$ \emph{crosses} the triple $\{S_1, S_2, S_3\}$ if
the boundary of $F$ contains at least one point $p_i$ from each set $S_i$ for $i = 1,2,3$ and
there is a cyclic traversal of the boundary of $F$ in which we see first $p_1$, then $p_2$ and finally $p_3$.
We write $]S_1,S_3[$ to denote the part of the boundary of $F$ that lies in between $p_1$ and $p_3$ and
contains $p_2$.

\begin{observation}\label{obs:onebend}
    Any fence in a fencing for an instance of \textsf{BMLF}
    crossing a non-collinear triple $\{S_1, S_2, S_3\}$ 
    contains at least one bend in the interval $]S_1, S_3[$.
\end{observation}

For $t > 0$ let $\mathcal S_1,\ldots,\mathcal S_t$ be non-collinear triples 
that are crossed by a fence $F$ of a fencing $\mathcal F$ for some instance of \textsf{BMLF}.
Let $\mathcal S_i = \{S_j,S_{j+1},S_{j+2}\}$ for $i = 1,\ldots,t$ and $j = 3(i-1) + 1$.
We say that the triples are crossed by $F$ \emph{in-order} if there exist points $p_i$ on $F$
such that $p_i$ is in the bend-set of $\mathcal S_i$ and there
exists a cyclic traversal of $F$ in which we see the points $p_i$ in order of their indices.
Without loss of generality we will assume throughout that when $F$ crosses $\mathcal S_1,\ldots,\mathcal S_t$ in-order
it always crosses for some $\mathcal S_i$ first the set $S_i$, then the bend-set $S_{i+1}$, and finally $S_{i+2}$.
We write $]S_a,S_b[$ with $a < b$ and $a = 1,\ldots,3t-1$ for the part of the boundary of $F$ that 
lies between a point $p_a \in S_a$ and $p_b\in S_b$ such that there exist points
$p_a,\ldots,p_b$ with $p_i \in S_i$ that we see in this order in a cyclic traversal of $F$.
For a segment $s$ of $F$ we say it is \emph{completely} contained in $]S_a,S_b[$ if
the start- and endpoint of $s$ are contained in $]S_a,S_b[$ for any choice of points $p_a$ and $p_b$.

\begin{figure}
    \centering
    \includegraphics[page=2]{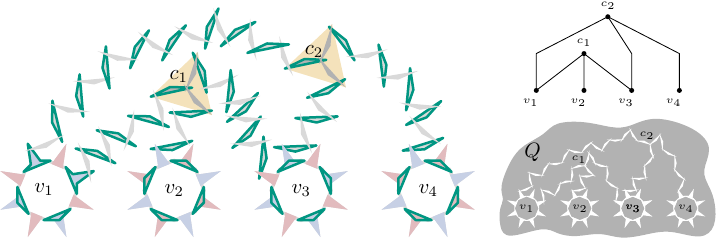}
    \caption{A (schematized) complete construction for a small instance $(v_1 \lor v_2 \lor \neg v_3) \land (v_1 \lor v_3 \lor v_4)$.
    The incidence graph is shown in the top right.
    Fences are highlighted in green.
    Also note that the boundary polygons make up most of the available area including an unbounded outer polygon $Q$ as shown in the bottom right corner.
    For better readability, we will invert these colors in all subsequent figures.
    }
    \label{fig:schematized_construction}
\end{figure}

We say two non-collinear triples 
$\mathcal S = \{S_1,S_2,S_3\}$ and $\mathcal S' = \{S_1',S_2',S_3'\}$ are \emph{non-overlapping} if
there exist no two segments that intersect all six elements of $\mathcal S \cup \mathcal S'$ in order
$S_1$, $S_2$, $S_3$, $S_1'$, $S_2'$, $S_3'$.
In other words we require at least three different straight-line segments 
to connect a point $p_1\in S_1$ with a point $p_6 \in S_3'$ and 
containing points $p_2\in S_2$, $p_3\in S_3$, $p_4\in S_1'$, and $p_5\in S_2'$ 
in order of their indices.
Observe that by this definition the non-collinear triples $\{S_1,S_2,S_3\}$ and
its reverse $\{S_3,S_2,S_1\}$ are non-overlapping.
For a sequence of non-collinear triples $\mathcal S_1,\ldots,\mathcal S_t$ we say that
the triples are non-overlapping if 
$S_i$ is non-overlapping with $S_{i+1}$ for $i = 1,\ldots,t+1 \mod t$.

Observation~\ref{obs:onebend} together with the definition of non-overlapping gives the following.
\begin{observation}\label{obs:kbends}
    Any fence in a fencing for an instance of \textsf{BMLF}
    crossing $t > 0$ non-overlapping non-collinear triples 
    $\mathcal S_i = \{S_j,S_{j+1},S_{j+2}\}$ for $i = 1,\ldots,t$ and $j = 3(i-1) + 1$ in-order 
    contains at least $t$ bends and therefore at least $t-1$ complete straight-line segments.
    in the interval $]S_1, S_{3t}[$.
\end{observation}

We can now show a lower bound for the number of links a minimum-link fence 
uses in any solution of a \textsf{BMLF} instance.
The lower bound essentially follows from Observation~\ref{obs:kbends} after 
observing that the segment closing the fence can never reuse one of 
the $t-1$ segments that lie completely inside the sequence of non-collinear triples.
\ifAppendixProofs
    \begin{restatable}{lemma}{lowerbound}\label{lem:lowerbound}
        Let $\mathcal P$ be an instance of \textsf{BMLF} and 
        $\mathcal F$ a minimum-link fencing for $\mathcal P$,
        then any fence $F \in \mathcal F$ that crosses $t > 0$ non-overlapping non-collinear triples in-order
        consists of at least $t$ straight-line segments.
    \end{restatable}
\else
    \begin{restatable}[$\star$]{lemma}{lowerbound}\label{lem:lowerbound}
        Let $\mathcal P$ be an instance of \textsf{BMLF} and 
        $\mathcal F$ a minimum-link fencing for $\mathcal P$,
        then any fence $F \in \mathcal F$ that crosses $t > 0$ non-overlapping non-collinear triples in-order
        consists of at least $t$ straight-line segments.
    \end{restatable}
\fi
\prooffromappendix{
\begin{proof}
    By Observation~\ref{obs:kbends} the fence $F$ has at least $t-1$ straight-line segments 
    Assume $F$ has also exactly $t-1$ segments.
    Let $\mathcal S_1,\ldots,\mathcal S_t$ be the non-overlapping non-collinear triples and
    $S_1,\ldots,S_{3t}$ be the sets of the non-collinear triple in the order in which $F$ crosses them as above.
    Since $F$ is a simple polygon there exists a polygonal chain $C$ starting at some point in $S_{3t}$ and
    ending at some point in $S_1$.
    By Observation~\ref{obs:kbends} there are $t-1$ straight-line segments 
    which are completely contained in $]S_1,S_{3t}[$ which implies
    that $C$ must consist of segments that are also completely contained in $]S_1,S_{3t}[$.

    Now, we charge every of the $t-1$ segments to the piece of $F$ that 
    connects for each $\mathcal S_i$ with $i = 1,\ldots,t-1$
    some point in $S_{i+1}$ with some point in $S_{i+2}$.
    Since the triples are non-overlapping each such connection also requires a distinct segment.
    Then, for the connection of $S_{3t-1}$ and $S_{3t}$ we require at least one more segment.
    Let $s$ be the segment charged to the connection of $S_{3t-4}$ to $S_{3t-3}$.
    If we could extend this segment to also intersect $S_{3t}$ we would violate
    that $\mathcal S_{t-1}$ and $\mathcal S_t$ are non-overlapping.
    Symmetrically for the segment charged to $S_2$ and $S_3$.
    Consequently, we require at least one more segment.
\end{proof}
}

\subsection{Variable gadget}\label{sec:variable}

\begin{figure}
    \begin{subfigure}[b]{0.45\linewidth}
        \centering
    	\includegraphics[scale=1.2,page=12]{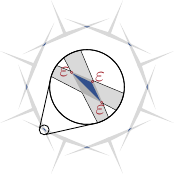}
        \caption{Variable gadget construction with $\varepsilon$-gaps.}
    	\label{fig:reduction1a}
	\end{subfigure}
    \centering
    \begin{subfigure}[b]{0.26\linewidth}
        \centering
        \includegraphics[page=10]{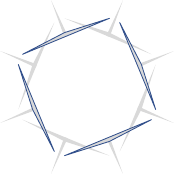}
        \subcaption{Fencing for \emph{true}-state.}
        \label{fig:variable_gadget_true}
    \end{subfigure}
    \begin{subfigure}[b]{0.26\linewidth}
        \centering
        \includegraphics[page=11]{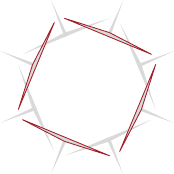}
        \subcaption{Fencing for \emph{false}-state.}
        \label{fig:variable_gadget_false}
    \end{subfigure}
    \caption{The variable gadget and its two possible fencings.}
    \label{fig:variable_gadget_fencings}
\end{figure}

Every variable gadget consists of eight \emph{T-polygons} \isaacrev{(two per clause in which the variable can occur)}.
\isaacrev{\Cref{fig:reduction1a} illustrates} the construction; T-polygons are marked in gray.
Every T-polygon
has an isosceles triangle as the \textit{arm} of the T (the
horizontal part of the T shape) and a \textit{spike} (alternatively called a \emph{true spike} and a \emph{false spike}) protruding from the arm and two consecutive polygons overlap at the end of their arms.
For every variable $v \in \mathcal{V}$, we construct a variable gadget $ \mathcal{G}(v) $ as a circular arrangement of eight overlapping T-polygons.
For every pair of overlapping T-polygons $A$ and $B$, we place an inner polygon $P$, s.t.\ $P \subset A\cap B$.
Let us fix some $A$, $B$, and $P$ as above,
then we place $P$ such that its three corner points 
have only a very small distance $\varepsilon>0$ 
to some corner point of $A \cap B$.
All three $\varepsilon$-length segments between a corner point of $P$ and the closest corner point of $A\cap B$ have to be crossed by every fence enclosing $P$.

Crucially, the variable gadget has only two minimum link fencings.
These two states are shown in Figure~\ref{fig:variable_gadget_fencings}.
We associate the one shown in Figure~\ref{fig:variable_gadget_true} to the variable gadget encoding 
the value \emph{true} and 
the one shown in Figure~\ref{fig:variable_gadget_false} to encoding \emph{false}.

\ifAppendixProofs
    \begin{restatable}{lemma}{vargadget}\label{lem:var_gadget}
    	There are exactly two minimum link fencings $\mathcal{F}_t$ and $\mathcal{F}_f$ of the variable gadget, 
    	both of which will enclose only triangles in the same T-polygon with each fence, 
    	resulting in a fencing with $ 12 $ links for the whole variable gadget, s.t. every other minimum link fencing is \congruent{} to either $\mathcal{F}_t$ or $\mathcal{F}_f$.
    \end{restatable}
\else
    \begin{restatable}[$\star$]{lemma}{vargadget}\label{lem:var_gadget}
    	There are exactly two minimum link fencings $\mathcal{F}_t$ and $\mathcal{F}_f$ of the variable gadget, 
    	both of which will enclose only triangles in the same T-polygon with each fence, 
    	resulting in a fencing with $ 12 $ links for the whole variable gadget, s.t. every other minimum link fencing is \congruent{} to either $\mathcal{F}_t$ or $\mathcal{F}_f$.
    \end{restatable}
\fi
\prooffromappendix{
\begin{proof}
    To prove this lemma, we will first number the eight inner polygons of the variable gadget $A_1, \ldots, A_8$.
    Any fence inside the variable gadget can include any combination of $k$ inner polygons.
    We will prove this lemma, by enumerating for every $k \in \{1, 8\}$ all possible (non-symmetric) combinations of including these inner polygons; Note that technically a fence is allowed to include non-consecutive inner polygons.
    Note that any fence has to consist of at least 3 segments and clearly every inner polygon of $\mathcal{G}(v)$ can be fenced alone with 3 segments.
    For any possible fence including multiple inner polygons, we show the existence of a certain number of non-overlapping non-collinear triples, which provide a lower bound on the number of segments for such a fence.
    All triples and the resulting lower bounds on the number of segments in fencings including all possible combinations of 2 to 8 inner polygons are shown in \ifArxiv Appendix~\ref{sec:app:variable_gadget}\else the full paper~\cite{}\fi.
    Enumeration shows that a fence $F$ including $k$ inner polygons, has the minimum amount of segments, if and only if the indices of the contained inner polygons are consecutive (assuming $A_8$ and $A_1$ to be consecutive).
    In particular, such a fence including 2, 3, 4, 5, 6, 7 or 8 inner polygons, requires (and can be realized with) 3, 6, 8, 11, 13, 16 and 18 segments, respectively.

    First observe that there are exactly two minimum link fencings, which include exactly two neighboring inner polygons in one fence, both consisting of 12 segments total, both are shown in Figures~\ref{fig:variable_gadget_true} and~\ref{fig:variable_gadget_false}.
    We can therefore exclude any fence including six or more inner polygons, since they clearly require more segments by themselves, which immediately eliminates the possibility of including 6 or more inner polygons in one fence.
    Any fence including five inner polygons requires at least 11 segments and at least one additional fence is needed, which increases the segment count to at least 14.
    Any fence including four inner polygons requires at least 8 segments. If the remaining four segments are fenced together, we require at least 16 segments. If they include a group of three polygons, we require at least 14 segments. If two remaining polygons are grouped, we require 11 segments, but at least one more fence (14 segments in total) are needed and if all remaining polygons are fenced alone, we need 20 segments in total.

    Next if three segments are fenced together, we require 6 segments and have five inner polygons left. We already know that no fence including four or more segments can be part of a minimum link fencing.
    If three of the remaining five polygons are grouped we arrive at a total of at least 12 segments, with at least one more fence needed.
    If at most two polygons are grouped, we need at least three more fences and arrive at a total of at least 15 segments.
    
    Finally, polygons could be fenced individually.
    Clearly there must be an even number of individually fenced polygons.
    If all eight polygons are fenced alone, we reach at least 24 segments, for six inner polygons fenced alone we get 18 segments, four individually fenced polygons lead to 12 with at least one more fence needed.
    Two inner polygons being fenced alone require at least 6 segments, with at least three more fences needed for the six left over polygons, which require at least 9 segments leading to a total of 15 segments.
    
    Therefore only two minimum link fencings exist, which require exactly 12 segments, and they group two neighboring inner polygons pairwise.
    We call the fencing, which groups inner polygons, which are both contained in a gray triangle with a true spike $\mathcal{F}_t$ and the other $\mathcal{F}_f$.
\end{proof}
}

\subsection{Clause gadget}\label{sec:clause}
For every clause $c \in \mathcal{C}$ in which three variables $v_1, v_2, v_3$ occur either as a positive or a negative literal, we create a clause gadget $\mathcal{G}(c)$.
A clause gadget consists of three chains of an even number of gray triangles.
These triangles are placed s.t.\ their hypotenuses intersect at an angle of at most $\pi$ as shown in \cref{fig:wire_construction-a}.
The triangles are sufficiently long and thin, s.t., we can define two sets in every gray triangle (one to either side of the central line), s.t., the second set $a'$ of the $i$-th triangle and the first set $b$ of the $(i+1)$-th triangle form a non-collinear triple.
By construction the non-collinear triple between the $(i-1)$-th and the $i$-th triangle and the one between the $i$-th and the $(i+1)$-th triangle are non-overlapping.

We place the three chains 
such that the three first triangles of the chains have a common intersection.
Moreover, they intersect in such a way that their hypotenuses pairwise form $\frac{2\pi}{3}$ angles (\cref{fig:wire_construction-b}).
The last gray triangle of the first, second and third chain intersect a spike of $\mathcal{G}(v_1)$, $\mathcal{G}(v_2)$ and $\mathcal{G}(v_3)$, respectively.
They intersect a true or false spike if the variable occurs as a positive or negative literal, respectively.
We refer to each chain of gray polygons as a \emph{wire}.
The \emph{length} of a wire is the number of gray triangles in its corresponding chain.

Let $W_1$, $W_2$, and $W_3$ be the wires of a clause gadget $\mathcal G(c)$ for clause $c$,
where $W_i$ intersects the spike of $\mathcal G(v_i)$ for $i \in \{1,2,3\}$.
We place an inner triangle, denoted the \emph{clause triangle} $B_c$ of $\mathcal G(c)$, 
in the overlap of $W_1$, $W_2$, and $W_3$.
Moreover, for wire $W_i$ with gray triangles $T^i_j$ we place inner triangles $B^i_j$
in the overlap of the $j$-th and $(j+1)$-th gray triangle of the respective wire and 
a final triangle in the intersection with the spike of $\mathcal G(v_i)$.
In the following we write $T_1,\ldots,T_k$ for the gray triangles and $B_1,\ldots,B_k$ for the inner polygons of one wire $W_i$,
if $i$ is clear from the context.
Hence, inner triangle $B_i$ is contained in the gray triangles $T_i$ and $T_{i+1}$ and
gray triangle $T_i$ for $i>1$ contains the inner triangles $B_i$ and $B_{i+1}$.

\begin{figure}[tbp]
    \centering
    \begin{subfigure}[b]{.6\linewidth}
        \centering
        \includegraphics[page=18]{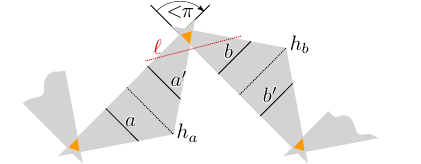}
        \subcaption{}\label{fig:wire_construction-a}
    \end{subfigure}
    \hfill
    \begin{subfigure}[b]{.38\linewidth}
        \centering
        \includegraphics[page=19]{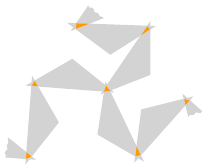}
        \subcaption{}\label{fig:wire_construction-b}
    \end{subfigure}
    \caption{Wires are constructed from consecutive gray triangles places such that two consecutive triangles always contain a non-collinear triple, which are pairwise non-overlapping.}
    \label{fig:wire_construction}
\end{figure}

Let $B_1,\ldots,B_k$ be the inner polygons of a wire and $F$ a fence 
containing $B_i$ and $B_j$ for some $i < j-1$ and $i = 1,\ldots,k-2$ but 
not $B_z$ for $i < z < j$, then we say $F$ \emph{bypasses} $B_{z}$.
For indices $1 \le i_1 < i_2 < j_1 < j_2 \le k$,
we say two fences $F_1$ and $F_2$ containing some polygons of the wire \emph{interleave} if
$B_{i_1}$ and $B_{j_1}$ are in $F_1$ and $F_1$ bypasses $B_{i_2}$ as well as
$B_{i_2}$ and $B_{j_2}$ are in $F_2$ and $F_1$ bypasses $B_{j_1}$.

Let $F$ be a fence of a minimum link fencing $\mathcal F$ for a clause gadget $\mathcal G(c)$.
Let $s$ be a segment contained in the union of the gray triangles that form $\mathcal G(c)$ such that
$F$ crosses $s$ in two points $p$ and $q$.
Then \emph{splitting} $F$ at $s$ means the following.
Delete $F$ in an $\varepsilon$-region around $p$ and $q$
this creates two polygonal-chains, say $F'$ and $F''$ with endpoints $p'$ and $q'$ on one side of $s$ 
and $p''$ and $q''$ on the other.
Connect $p'$ with $q'$ and $p''$ with $q''$ to form the two new fences $F'$ and $F''$.
Clearly, $|F'| + |F''| = |F| + 2$.

\paragraph*{One isolated wire}
For the following we fix an arbitrary clause $c$.
Let $\mathcal G(c)$ be the clause gadget of $c$ and $W$ one of the wires of $\mathcal G(c)$
with inner polygons $B_1,\ldots,B_k$.
We denote as \emph{isolated wire} the gray triangles of the chain of $W$
that do not contain the clause triangle.

We are interested in how a minimum link fencing of an isolated wire looks like.
Crucially, we first show that a fence of a minimum link fencing of an isolated wire
cannot bypass any inner polygon of an inner polygon. 

\begin{restatable}[$\star$]{lemma}{consecutive}\label{lem:consecutive}
    A minimum link fence $\mathcal F$ of an isolated wire $W$ of $\mathcal G(c)$
    does not contain a fence $F \in \mathcal F$ such that $F$ bypasses 
    an inner polygon $B_i$ with $i \in \{2,\ldots,k-1\}$ of $W$.
\end{restatable}

In the following we are going to bound the number of consecutive polygons that are contained in one 
minimum link fence of an isolated wire.
We compare this then to a fence containing all inner polygons of an isolated wire.
Such a fence, by construction, contains $2z$ non-collinear triples and hence requires $2z$ segments by Lemma~\ref{lem:lowerbound}.
Figure~\ref{fig:non_collinear_triples_wire} shows these triples.
Constructing such a fence is straight-forward by following these non-collinear triples.
The following lemma summarizes this statement.

\begin{restatable}{lemma}{conspolys}\label{lem:conspolys}
    Let $\mathcal F$ be a minimum link fencing of an isolated wire $W$ of $\mathcal G(c)$, 
    any fence $F \in \mathcal F$ that contains $z>2$ consecutive inner polygons of $W$ 
    has at least $2z$ segments and such a fence exists.
\end{restatable}

\begin{restatable}{lemma}{sizebound}\label{lem:sizebound}
    Let $\mathcal F$ be a minimum link fencing of an isolated wire $W$ of $\mathcal G(c)$, then 
    every fence of $\mathcal F$ contains at most three consecutive inner polygons.
\end{restatable}
\begin{proof}
    Let $B_1,\ldots,B_k$ be the inner polygons of $W$.
    By Lemma~\ref{lem:consecutive} we can assume that the inner polygons of $W$ contained in $F$ are
    consecutive in the sequence of inner polygons.
    Let $F \in \mathcal F$ be a fence containing $z > 3$ inner polygons of $W$.
    
    By Lemma~\ref{lem:conspolys} we know that $F$ consists of $2z$ segments.
    We replace $F$ with a fence $F_1$ including the two first polygons included in $F$ and
    a fence $F_2$ including all $z-2$ following inner polygons.
    Again by Lemma~\ref{lem:conspolys} it follows that $|F_2| = 2z - 4$ and
    it holds that $|F_1| = 3$.
    In sum, we get that $|F_1| + |F_2| = 2z - 4 + 3 = 2z - 1 \leq |F|$,
    a contradiction.
\end{proof}

\begin{figure}[tbp]
    \centering
    \includegraphics[page=2]{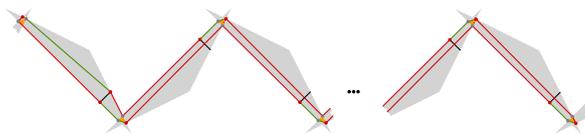}
    \caption{Non-colinear triples in a fence including a series of consecutive inner polygons.}
    \label{fig:non_collinear_triples_wire}
\end{figure}

\cref{lem:sizebound,lem:conspolys} now lead to 
a characterization of minimum link fences of isolated wires.

\begin{restatable}[$\star$]{lemma}{onlytwofences}\label{lem:onlytwofences}
    Let $\mathcal F$ be a minimum link fencing of an isolated wire $W$ of $\mathcal G(c)$ with 
    $k$ inner polygons,
    then $\mathcal F$ has in total $3k/2$ segments and
    $F \in \mathcal F$ contains exactly two consecutive inner polygons $B_i$ and $B_{i+1}$ for $i$ odd.
\end{restatable}

\paragraph*{Integrating the clause triangle}
So far we only considered one arbitrary isolated wire of $\mathcal G(c)$.
To put things together we need to consider the interaction of the three wires of $\mathcal G(c)$.
Specifically, we need to show that no fence in a minimum link fencing of $\mathcal G(c)$
contains inner polygons from two different wires.

We extend the definition of bypassing an inner polygon of a wire to a whole clause gadget.
Let $F$ be a fence for $\mathcal G(c)$, 
then $F$ \emph{bypasses} an inner polygon $B_j^i$ of wire $W_i$ of $\mathcal G(c)$ if 
$F$ contains the clause triangle $B_c$ or some inner polygon of a wire $W_{i'}$ with $i' \neq i$ and 
$F$ contains $B_l^i$ for wire $W_i$ with $l > j$.
We say $F$ \emph{bypasses} the clause triangle of $\mathcal G(c)$ if 
$F$ contains inner polygons of at least two different wires of $\mathcal G(c)$ but 
not the clause triangle $B_c$ of $\mathcal G(c)$.

As for an isolated wire we can show that no inner polygon for a whole clause gadget can be bypassed.
This can be seen after observing that no fence can bypass the inner polygons of an isolated wire 
without violating Lemma~\ref{lem:consecutive}.
The remainder of the proof is then a careful case enumeration, \ifArxiv see Appendix~\ref{apx:nphard:clause}\else which can be found in the full paper~\cite{}\fi.

\begin{restatable}[$\star$]{lemma}{nobypassinginclause}\label{lem:no_bypassing_in_clause}
    Let $\mathcal F$ be a minimum link fencing of $\mathcal G(c)$ and
    $B_1,\ldots,B_k$ the inner polygons of one of the wires of $\mathcal G(c)$.
    Then there is no fence $F\in\mathcal F$ that bypasses an inner polygon $B_i$ with $i \in \{1,\ldots,k-1\}$.
\end{restatable}

Finally, we show that no minimum link fence of a clause gadget can ever fence 
two polygons that are in different wires.
Again, this is shown essentially via a case enumeration that considers
how a minimum link fence includes the first two to three polygons of each wire together with the clause triangle.
In each case we can conclude that there exists a fence with fewer segments that in fact does
not use the inner polygons of two distinct wires.

\begin{restatable}[$\star$]{lemma}{onlyinonewire}\label{lem:only_in_one_wire}
    Let $\mathcal F$ be an optimal fencing of a clause gadget, then there exists no fence $F \in \mathcal{F}$, which includes inner polygons belonging to two different wires.
\end{restatable}

We can now use Lemma~\ref{lem:only_in_one_wire} to argue that the clause triangle is only included in a fence together with inner polygons of at most one wire.
\soerenchange{We say that such a wire is in a \emph{satisfying} state.}
The other two wires should therefore, by Lemma~\ref{lem:onlytwofences}, only use fences including two inner polygons; leading to $\frac{3(k_a + k_b + k_c)}{2} + 3$ segments in total ($k_a, k_b$ and $k_c$ being the number of inner polygons in the wires).
If we include the clause triangle in a fence of a wire, we get the same amount of segments, however, we can choose fences, s.t., the last inner polygon of the wire which fences the clause triangle, is fenced alone.
This will be crucial in the argument of how the wires and therefore the clause gadget interacts with the variable gadget.

\paragraph*{Interaction with the variable gadgets}
It remains to describe the interaction between the variable and clause gadgets.
Depending on the state of the variable gadget we can fence the last inner polygon of a wire
in the fence of a variable gadget.
We provide a fence with 5 segments (i.e., two additional ones) for the case, where the variable gadget is in the correct state and the existence of six non-collinear triples for the other case, see Figure~\ref{fig:wire}.%

\begin{restatable}{lemma}{variablecost}\label{lem:variablecost}
    The last inner polygon of a wire can be included in a fence of the variable gadget, whose spike it is connected to for the cost of two additional segments if the variable gadget is in the correct state and at least three additional segments otherwise.
\end{restatable}
    
\begin{figure}
    	\centering
    	\begin{minipage}[t]{.45\textwidth}
    		\centering
    		\includegraphics[page=5, width=\linewidth]{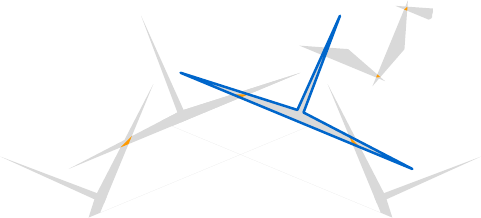}
    		\subcaption{$ \mathcal{G}(v) $ in correct state}
    		\label{fig:reduction_wire_truth}
    	\end{minipage}
    	\hfill
    	\begin{minipage}[t]{.45\textwidth}
    		\centering
    		\includegraphics[page=4, width=\linewidth]{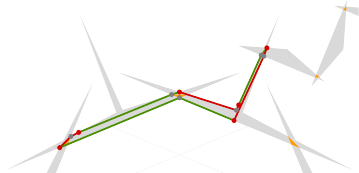}
    		\subcaption{$ \mathcal{G}(v) $ in incorrect state}
    		\label{fig:reduction_wire_truth_a}
    	\end{minipage}
    	\caption{If $\mathcal{G}(v)$ is in the correct truth state (a) inclusion of the first inner polygon of 
    		$\mathcal{G}(v, c)$ induces an additional cost of two links,
    		otherwise (b) the additional cost is at least three.}\label{fig:wire}
    \end{figure}

Concluding the interaction between clause and variable gadget we show that
given a variable gadget is in the correct state w.r.t.~a clause gadget 
we can fence the inner polygons of the wires of a clause gadget using $3/2$ segments per polygon
and adding only two segments to the fence of the variable gadget. 

\begin{restatable}[$\star$]{lemma}{clausecost}\label{lem:clausecost}
    If and only if at least one of the connected variable gadgets is in the correct state, the clause gadget can be fenced with a total of $\frac{3(k_a + k_b + k_c)}{2}$ segments plus two additional segments to a fence of the variable gadget, which is connected to the wire in the satisfying state.
\end{restatable}

\paragraph*{Correctness}

It remains to argue the correctness of our reduction which then implies our main theorem.

\begin{restatable}[$\star$]{theorem}{nphard}\label{thm:nphard}
    Two-colored \textsf{BMLF} is \NP-hard even when restricting all fences to include at most three polygons.
\end{restatable}
\soerenchange{
\begin{proof}[Proof sketch]
    For an instance $\phi$ of planar $3,4$-\SAT, we construct a variable gadget for every variable and connect the clause gadgets accordingly.
    By construction any fencing with $|\mathcal V| \cdot 12 + \sum_{c \in \mathcal C} (\frac{3(k^c)}{2} + 2)$ segments, requires one wire of every clause gadget to be in a satisfying state.
    The connected variable gadget is forced into the true or false state, depending on the connected spike.
    This implies a satisfying variable assignment for $\phi$.
    
    Conversely, since every variable is either true or false and for every clause there is a true literal, we can set all variable gadgets into the true or false state according to the assignment and are guaranteed to be able to put exactly one wire per clause gadget into a satisfying state for an additional cost of exactly two.
\end{proof}}

\section{An \XP-algorithm for \textsf{BMLF} with at most two polygons in each fence}
\label{sec:XPalgo}

In \cref{sec:nphard} we showed that \textsf{BMLF} is \NP-hard when
there are only two colors, each fence contains at most three polygons, and 
each fence consists of at most five links.
In contrast, we are going to show in this section that \textsf{BMLF} can be solved in \XP-time when 
parameterizing the problem by the maximum number of links in any fence and allowing
at most two polygons per fence, i.e.,
the problem can be solved in polynomial-time when fixing the maximum 
number of links in any fence and restricting each fence to contain at most two polygons.

For our algorithm we make use of the following result derived from the work of Hershberger and Snoeyink%
~\cite{DBLP:journals/comgeo/HershbergerS94}.
It allows us to compute for a given \emph{loop}, i.e., a closed polygonal curve, inside a polygon with holes, a minimum-link loop of the same homotopy
in time $O(nk)$, where $n$ is the complexity of the polygon and $k$ is the size of the resulting fence.

\begin{theorem}[Derived from Section 5.2~\cite{DBLP:journals/comgeo/HershbergerS94}]
    \label{thm:klinkblackbox}
    Given a polygon $P$ without self-intersections but potentially with holes of
    complexity $n$,
    an integer $k$, and
    a loop $\alpha$ lying in the interior of $P$ with $O(nk)$ corners,
    we can decide in time $O(nk)$ if there exists a loop $\alpha'$
    of the same homotopy-class as $\alpha$ with at most $k$ links.

\end{theorem}

\begin{remarkenv}
It is worth noting that in the paper by Hershberger and Snoeyink~\cite{DBLP:journals/comgeo/HershbergerS94} 
    Theorem~\ref{thm:klinkblackbox} is only stated in text.
    The runtime is given as $O(C_\alpha + \Delta_\alpha + \Delta_{\alpha'})$, where 
    $C_\alpha$ is the complexity of $\alpha$, the free space between polygons is assumed to be triangulated and $\Delta_\alpha$ and $\Delta_{\alpha'}$ are the number of triangulation edges intersected by $\alpha$ and the fence $\alpha'$, respectively.
    However an example of an instance with multiple obstacles is given, in which $\Delta_{\alpha'}\in\Omega(nk)$, where $n$ is the number of corners over all polygons.
    Since in our scenario we can find a path $\alpha$ s.t. $C_\alpha \in O(nk)$ and $\Delta_\alpha \in O(nk)$, we can make the assumption that $\alpha'$'s complexity is 
    in $O(nk)$.
\end{remarkenv}

Let $P$ be a polygon without self-intersections.
We denote with $\mathcal T_P = \{T_1,\ldots,T_z\}$ a triangulation of $P$ 
with triangles $T_1,\ldots,T_z$.
Note that we do not require any further properties of $\mathcal T_P$.
If $P$ is clear from the context we omit it and set $\mathcal T = \mathcal T_P$.
Let $T_1,T_2 \in \mathcal T$ be two triangles and let 
$l$ be a line segment with endpoints $p$ and $q$ such that $p \in T_1$ and $q \in T_2$.
We call $l$ a \emph{splitting segment}.
Consider Figure~\ref{fig:bitangents_a} for an example
for $T_1$ and $T_2$ if
$l$ contains no points of $\mathbb R^2 \setminus P$.
Intuitively, a splitting segment separates the holes that intersect the convex hull of $T_1 \cup T_2$
into two sets.
Let $\mathcal H$ be all the holes of $P$ that intersect or are fully contained 
in the interior of the convex hull of $T_1 \cup T_2$.
We say that a hole $H \in \mathcal H$ is
to the \emph{left} (\emph{right}) of $l$ if the from $p$ to $q$ oriented supporting line of $l$
leaves $H$ in the left (right) half-plane.
We call two splitting segments of $T_1$ and $T_2$ \emph{equivalent}
if the same holes of $\mathcal H$ are to their respective left and right. Segments which intersect holes are not splitting segments.

\ifAppendixProofs
    \begin{restatable}{lemma}{equivalenceclasses}
    \label{lem:equivalence_classes}
        Let $P$ be a polygon without self-intersection, $\mathcal H$ a set of holes and 
        $\mathcal T$ a triangulation of $P$.
        Then for every pair of triangles $T_1,T_2 \in \mathcal T$ with $T_1 \neq T_2$
        there are at most $4|\mathcal H|^2$ different equivalence classes of splitting segments.
    \end{restatable}
\else
    \begin{restatable}[$\star$]{lemma}{equivalenceclasses}
    \label{lem:equivalence_classes}
        Let $P$ be a polygon without self-intersection, $\mathcal H$ a set of holes and 
        $\mathcal T$ a triangulation of $P$.
        Then for every pair of triangles $T_1,T_2 \in \mathcal T$ with $T_1 \neq T_2$
        there are at most $4|\mathcal H|^2$ different equivalence classes of splitting segments.
    \end{restatable}
\fi
\prooffromappendix{
\begin{proof}
    Let $\mathcal H^\cap \subseteq \mathcal H$ be the set of holes, which intersect the joint convex hull $C$ of $T_1$ and $T_2$ or
    are fully contained in it.
    Let $l$ be a splitting  segment with endpoints $p \in T_1$ and $q \in T_2$.
    Observe, that since $l$ is completely contained in $C - \mathcal H$ we,
    by definition of equivalent splitting segments, can disregard any hole $H \not \in \mathcal H^\cap$.

    Let $l \in C$ be a splitting segment of $T_1$ and $T_2$, 
    which has a set $L \subseteq \mathcal H^\cap$ of holes to its left and
    a set $R \subseteq \mathcal H^\cap$ to its right.
    There exists another splitting segment $l_c$ 
    whose supporting line is a bitangent on two holes, say $H_L \in L$ and $H_R \in R$,
    such that $l_c$ has also $L$ to its left and $R$ to its right.
    We obtain the splitting segment $l_c$ from $l$  by
    rotating $l$ first around its center point
    until it touches either $H_L$ or $H_R$ in a point $r$ and 
    then continue to rotate $l$ around $r$ until it touches the second hole, see \cref{fig:bitangents_a}. 
    Since any pair of holes has at most four bitangents (\cref{fig:bitangents_b}) there are at most 
    $4|\mathcal H^\cap|^2$ equivalence classes.
\end{proof}
}

    \begin{figure}
        \centering
        \begin{subfigure}[t]{.7\linewidth}
            \centering
            \includegraphics[page=4]{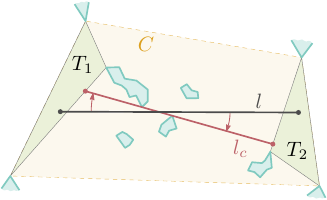}
            \subcaption{A splitting segment $l_c$ equivalent to $l$ and bitangent to two holes.}
            \label{fig:bitangents_a}
        \end{subfigure}
        \qquad
        \begin{subfigure}[t]{.24\linewidth}
            \centering
            \includegraphics[page=2, width=.6\linewidth]{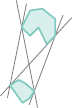}
            \subcaption{Bitangents}
            \label{fig:bitangents_b}
        \end{subfigure}
        \caption{
        \sublab{a} We can obtain the splitting segment $l_c$ from a splitting segment $l$ by rotating $l$ clockwise until it is a bitangent to two holes.
        \sublab{b} Any pair of two polygons admits at most four bitangents, only one of which can not be rotated clockwise without intersecting one of the polygons.
        }
        \label{fig:bitangents}
    \end{figure}

\ifAppendixProofs
    \begin{restatable}{theorem}{xpalgorithm}
        Given an instance $\mathcal P$ of \textsf{BMLF} with outer polygon $Q \in \mathcal P$,
        we can decide in time $O(kn^{2k+4})$ if a minimum link fencing $\mathcal{F}$ of $\mathcal{P}$ 
        exists, in which every fence contains at most two polygons, each fence in the fencing has at most $k$ segments, and
        $n$ is the number of corners in $\mathcal P$.
    \end{restatable}
\else
    \begin{restatable}[$\star$]{theorem}{xpalgorithm}
        Given an instance $\mathcal P$ of \textsf{BMLF} with outer polygon $Q \in \mathcal P$,
        we can decide in time $O(kn^{2k+4})$ if a minimum link fencing $\mathcal{F}$ of $\mathcal{P}$ 
        exists, in which every fence contains at most two polygons, each fence in the fencing has at most $k$ segments, and
        $n$ is the number of corners in $\mathcal P$.
    \end{restatable}
    \begin{sketch}
    For each polygon and for each pair of polygons we compute a minimum-link fence.
    Let $\lambda_{uv}$ be the number of links for a minimum link fence containing $P_u,P_v\in\mathcal P$ and
    $\lambda_u$ the number of links for a minimum link fence containing only $P_u \in\mathcal P$.
    Consider a complete graph $G$ containing one vertex $u$ for each polygon $P_u\in\mathcal P$ and
    one more vertex $x$ if $|\mathcal P|$ is odd.
    Set the edge-weights $w(u,v) = \min\{\lambda_{uv},\lambda_u + \lambda_v\}$ and
    $w(x,u) = \lambda_u$ for $P_u,P_v \in \mathcal P$.
    If for some $P_u \in \mathcal P$ or pair $P_u,P_v \in \mathcal P$ 
    no fence with $\leq k$ segments exists we remove that edge.
    Computing a minimum weight perfect matching in this graph yields a minimum link fencing.

    It remains to compute the minimum-link fences for each polygon and for each pair of polygons of $\mathcal P$.
    We consider a triangulation $\mathcal T$ of the free space of $\mathcal P$.
    For one single polygon we construct a plane loop around it by just traversing the incident triangles in the triangulation.
    To compute the minimum link fence for a pair of polygons in $\mathcal P$ we need to do more work.
    Since $\mathcal T$  contains only $O(n)$ triangles we can
    branch over the $O(n^k)$ ordered $k$-tuples of triangles.
    Moreover, by Lemma~\ref{lem:equivalence_classes} we can branch over the $O(n^{2k})$
    different splitting segments.
    If for our choice of triangles all splitting segments between consecutive triangles exist
    we construct a plane loop $\alpha$ if possible or otherwise reject the branch.
    
    Let $T_1,\ldots,T_k$ be the chosen $k$-tuple of triangles and $l_1,\ldots,l_k$ the splitting segments.
    If none of the splitting segments intersect 
    the same triangles in between two consecutive triangles $T_i$ and $T_{i+1}$ this is straight forward.
    If there are triangles that are intersected multiple times we have to evaluate $2^{O(k)}$ choices
    of how to resolve the self-crossings such a repetition induces for the loop $\alpha$.
    For each valid choice we apply Theorem~\ref{thm:klinkblackbox}.
    
    These are only $O(kn^{2k+4})$ choices in total and 
    computing a minimum weight perfect matching can be done in $O(V^2 E)$ time (with 
    $V$ being the number of vertices and $E$ the number of edges) via finding a maximum weight perfect matching (e.g.~\cite{edmonds_1965}) on the same graph with edge weights set to
    maximum edge-weight plus one minus the original edge-weight.
    \end{sketch}
\fi
\prooffromappendix{
\begin{proof}
    Throughout, we consider the triangulation $\mathcal T$ of the free space of $\mathcal P$.
    Let $t+1$ be the number of colors in the given instance and
    $c_1, \ldots, c_t$ be the colors of polygons in $\mathcal P$ with $c_i \neq f(Q)$ for all $i = 1,\ldots,t$.
    Observe that the homotopy of a fence including exactly one polygon $P \in \mathcal P$ is unique.
    We find a path $\alpha$ with this homotopy in the triangulated free space between the polygons in $\mathcal P$ by traversing the boundary of $P$ clockwise and at every corner of the polygon adding every incident triangle of $\mathcal T$ in clockwise order to a list.
    This yields a series of triangles from which we can construct a loop $\alpha$ 
    such that $P$ is contained in $\alpha$.
    This can be done, for example, by connecting all midpoints of triangulation segments of a triangle and its successor in the loop.
    Then, we use \cref{thm:klinkblackbox} to obtain a minimum-link fence for $P$ from $\alpha$
    or determine that no fence with at most $k$ links exists.
    Computing all individual fences requires $O(|\mathcal{P}|kn) = O(kn^2)$ time.
    
    Next, we consider every pair $P_1,P_2 \in \mathcal P$ with 
    $P_1 \neq P_2$, $c_i = f(P_1) = f(P_2)$ and $i \in \{1,\ldots,t\}$.
    Note that in contrast to a fence containing only one polygon, 
    a fence containing exactly two polygons can belong to several different homotopy-classes w.r.t.
    the remaining polygons. 
    
    Now, we describe how to compute a minimum-link fence $F$ for two polygons $P_1$ and $P_2$ of $\mathcal P$.
    Recall that all polygons in $\mathcal P$ together have $n$ corners.
    Since we have only $O(n)$ triangles in $\mathcal T$
    we can iterate over all possible $O({n \choose k})\in O(n^k)$ ordered tuples of $k$ triangles.
    Fix in the following such an ordered $k$-tuple $(T_1,\ldots,T_k)$ of triangles in $\mathcal T$.
    There are only $O(|\mathcal P|^2)$ many non-equivalent splitting segments 
    connecting points in triangles $T_i$ and $T_{i+1}$ by Lemma~\ref{lem:equivalence_classes}.
    Consequently, we can iterate over the $O((|\mathcal P|^2)^k) = O(n^{2k})$ 
    many different combinations of splitting segments
    between consecutive triangles.
    In case there are two consecutive triangles between which no possible splitting line exists we reject
    this tuple of triangles.
    Assume in the following that we fix for every pair $T_i$ and $T_{i+1}$ a splitting segment $l_i$.

    It remains to construct a plane loop $\alpha$ as input for the algorithm of Hershberger and Snoeyink~\cite{DBLP:journals/comgeo/HershbergerS94} or
    decide that no such loop exists for the fixed choices of triangles and splitting segments.
    From the triangles $T_1,\ldots,T_k$ and the splitting segments $l_1,\ldots,l_k$
    we derive a sequence of triangles $\tau_1,\ldots,\tau_z$ of $\mathcal T$ that $\alpha$ has to visit.
    Since the triangulation $\mathcal T$ is defined by the corners of polygons in $\mathcal P$
    each splitting segment $l_i$ gives rise to a unique sequence of triangles.
    We concatenate all these sequences starting with the sequence induced by $l_1$ 
    to obtain the sequence $\Theta = (\tau_1,\ldots,\tau_z)$.
    Observe, that triangles along this sequence may repeat and that $\tau_1 = \tau_z$.
    
    It remains to decide if there exists a plane loop $\alpha$ visiting each triangle of $\Theta$ in order.
    To make the following description simpler let $s_i$ be the shared boundary of $\tau_i$ and $\tau_{i+1}$.
    If for no $i$ with $i \in \{1,\ldots,z-1\}$, we find that $s_i = s_{i+1}$ we create a loop $\alpha$ by
    connecting the centerpoint of $s_i$ with the one of $s_{i+1}$.
    Since no boundary repeats, this is always possible without any centerpoint and 
    hence triangle being used twice.
    Finally, we add the segment from the centerpoint of $s_z$ to the centerpoint of $s_1$ which 
    is also always possible as $\tau_1 = \tau_z$.
    
    Now assume there exist at least two indices $i$ and $j$ with 
    $i \neq j$ and $i,j\in\{1,\ldots,z\}$ such that $s_i = s_j$.
    Build the loop as before and let $\alpha_1,\ldots,\alpha_z$ be the segments in the constructed loop.
    Since we allow repeated boundaries there exist subsequences 
    among the $\alpha_1,\ldots,\alpha_z$ that are repeated.
    In the following we assume that we only consider inclusion maximal repeated subsequences.
    Let $A = \{\alpha_1',\ldots,\alpha_a'\}$ be one occurrence of such a subsequence of the $\alpha_i$'s
    that is repeated at least once and let $\hat A$ be a different occurrence.
    Let $s_1',\ldots,s_a'$ be the subsequence of the $s_i$'s that correspond to the 
    triangle boundaries passed by the segments in $A$ and $\hat A = \{\hat{\alpha_1},\ldots,\hat{\alpha_a}\}$.
    Now observe that in a plane loop the vertices of $A$ and $\hat A$ have to always appear 
    in the same order along $s_1',\ldots,s_a'$.
    If they would not, let $s_i'$ and $s_{i+1}'$ for $i \in \{1,\ldots,a-1\}$ be two segments
    such that the vertices of $A$ and $\hat A$ on $s_i'$ and $s_{i+1}'$ are not in the same order.
    Without loss of generality assume the vertex of $A$ on $s_i'$ is above the one of $\hat A$ and %
    the opposite is true for $s_{i+1}'$, 
    then we find that $\alpha_i$ and $\hat{\alpha_i}$ cross.
    Hence, the only decision to make is to decide, for each pair of repeated subsequences, in which order 
    their vertices appear along the corresponding triangle-boundaries.
    Since every repeated subsequence implies an intersection 
    between two segments $l_i$ and $l_j$ with $i \neq j$ and $i,j\in \{1,\ldots,k\}$
    we find that there are at most $k^2$ such repeated sequences.
    Consequently, there are at most $2^{O(k)}$ possible ways to distribute
    the center points in each shared part along the boundaries.

    To sum up, for one pair of polygons $P_1,P_2$ we have to consider 
    $O(n^{2k})$ possible non-homotopy equivalent fences and
    for each homotopy we can check in $O(n^2 \cdot 2^k)$ 
    if there exists a plane loop $\alpha$ realizing it, leading to a total runtime of $O(2^k\cdot n^{2k+2})$
    to enumerate every potential homotopy of a minimum link fence.
    For each of the $O(n^{2k})$ different homotopies, we can use \cref{thm:klinkblackbox} to compute a minimum link fence in $O(kn)$, hence we can compute a minimum link fences for all pairs of polygons in $O(n^2\cdot kn^{2k+2}) = O(kn^{2k+4})$.

    If for any polygon no fence, alone or in a pair with another polygon, 
    with $k$ or fewer links is found, 
    we return that no solution exists.
    Otherwise, let $\lambda_{uv}$ be the number of links for a minimum link fence containing $P_u,P_v\in\mathcal P$ and
    $\lambda_u$ the number of links for a minimum link fence containing only $P_u \in\mathcal P$.
    Consider a complete graph $G$ containing one vertex $u$ for each polygon $P_u\in\mathcal P$ and
    one more vertex $x$ if $|\mathcal P|$ is odd.
    Set the edge-weights $w(u,v) = \min\{\lambda_{uv},\lambda_u + \lambda_v\}$ and
    $w(x,u) = \lambda_u$ for $P_u,P_v \in \mathcal P$.
    If for some $P_u \in \mathcal P$ or pair $P_u,P_v \in \mathcal P$ 
    no fence with $\leq k$ segments existed we remove that edge.

    To find a minimum-link fencing of $\mathcal P$ it now suffices to compute a minimum weight perfect matching in $G$.
    Let $M$ be such a matching.
    Then, a minimum link fencing $\mathcal{F}$ of $\mathcal{P}$ can be constructed from $M$ in the following way.
    If $uv \in M$, we add the (pre-computed) minimum link fences containing only $P_u$ and $P_v$ to $\mathcal{F}$ if the weight $w(u,v) = \lambda_u + \lambda_v$ or
    the fence containing $P_u$ and $P_v$ if $w(u,v) = \lambda_{uv}$.
    If $|\mathcal P|$ was odd we also find an edge $xu \in M$ and
    we add the fence containing only $P_u$ to the fencing.
    Finding a minimum weight perfect matching in a general graph with $V$ vertices and $E$ edges can be done for example in $O(V^2E)$ time 
    via finding a maximum weight perfect matching (e.g.~\cite{edmonds_1965}) in the same graph with edge weights set to
    maximum edge-weight plus one minus the original edge-weight.
    Since $G$ has $O(|\mathcal{P}|)$ vertices and $O(|\mathcal{P}|^2)$ edges we can compute this matching in $O(|\mathcal{P}|^4) = O(n^4)$, which is dominated by the initial computation of the minimum link fences.
\end{proof}
}

\section{An algorithm for two-colored \textsf{CMLF}}
\label{sec:algorithm}

In this section we present an algorithm for solving two-colored \textsf{CMLF}. 
Computing a minimum-link fence in this setting can be done by computing a fence for the convex hull of the contained polygons with the algorithm by Wang~\cite{wang_finding_1991}
which runs in time $O(n\log n)$ with $n$ being the number of corners of the contained polygons.
Throughout this section an instance of \textsf{CMLF} is given as $(\mathcal P, Q)$
where $Q$ is the outer polygon and $\mathcal P$ is the set of polygons contained in $Q$.

\begin{figure}[tb]
	\centering
	\includegraphics[page=2]{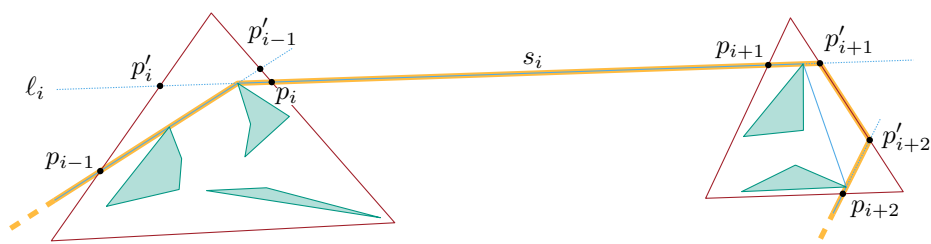}
	\caption{
	Computing a new fence (orange) from the old fences (purple) and the convex hull~(blue).}
	\label{fig:convex_hull_lma}
\end{figure}

\begin{lemma}\label{lmachull}
    Given an instance $(\mathcal{P},Q)$ of two-colored \textsf{CMLF}, let
    $\mathcal F$ be a solution for $(\mathcal P, Q)$.
    There exists a solution $\mathcal F'$ for the two-colored \textsf{CMLF} instance $(CH(\mathcal P),Q)$ with $|F| = |F'|$.
\end{lemma}

\begin{proof}
As $F$ is a minimum-link fencing of $(\mathcal P, Q)$,
it suffices to consider the case 
\isaacrev{where a minimal link fencing of $(CH(\mathcal P),Q)$ has strictly more segments than $|F|$.}
We will construct a new fence $F^\circ$ from this instance.
Let $(p_1,\dots,p_z)$ be the intersection points between $F$ and $CH(\mathcal P)$ ordered as they appear in a clockwise traversal of the convex hull, and observe that $z$ is even.
Let $p_i$, $p_{i+1}$ be pairs of intersection points between $F$ and $CH(\mathcal P)$ such that the straight-line segment $s_i$ connecting
$p_i$ and $p_{i+1}$ lies on $CH(\mathcal P)$ and completely outside of $F$ (see \cref{fig:convex_hull_lma}).
Consider the supporting line $\ell_i$ of $s_i$.
If the fence lies completely in one of the closed half-planes bounded by $\ell_i$
we add $s_i$ to $F^\circ$.
Assume this is not the case.
As $s_i$ is on 
$CH(\mathcal P)$ we get that $\ell_i$ does not intersect any polygon in $\mathcal P$.
Moreover, as $\mathcal F$ consists of closed simple polygons we find two intersection points $p_i'$ and $p_{i+1}'$ that lie on $\ell_i$, s.t.,
the parts of $F$ appearing in a \isaacrev{clockwise traversal from $p_i'$ to $p_i$}, as well as the ones in a clockwise traversal from $p_{i+1}$ to $p_{i+1}'$ lie outside of $CH(\mathcal P)$.
We add the segment $s_i'$ between $p_i'$ and $p_{i+1}'$ to $F^\circ$.
Doing this for every pair of intersections we obtain a set of segments $F^\circ$, where all segments are on the convex-hull of $\mathcal P$.
Note that it is possible for these segments to intersect; if that is the case
we only keep the parts until their intersection point.
Finally, the start and end-points of connected chains of segments in $F^\circ$
lie on segments of \isaacrev{fences}
in $\mathcal F$.
We can convert $F^\circ$ into a fence of $CH(\mathcal P)$ by connecting these endpoints along the \isaacrev{fences }%
in $\mathcal F$ and that fence will be disjoint from $\mathcal P$ (except possibly touching $\mathcal P$ in corner points).

It remains to argue that indeed $|F^\circ| \leq |F|$.
We partition $F^\circ$ into two categories,
segments that coincide with segments in $F$ and segments that do not.
Each of them is either a full segment of $F$ or 
originates from the intersection of at most two different $s_i'$'s and a segment of $F$.
Furthermore, we add $z/2$ segments $s_i'$ that are not sub-segments of segments in $F$. 
For each such $s_i'$ we find at least one segment of $F$ for which we did not add any sub-segment to $F^\circ$.
These are the segments of $F$ on which $p_i$ and $p_{i+1}$ lie or
that are fully outside of $F^\circ$.
\end{proof}

\begin{theorem}
    \label{thchull}
    \isaacrev{Two-colored \textsf{CMLF}} can be solved in time $O(n\log n)$ where $n$ is the number of corners of polygons in $\mathcal P$.
\end{theorem}

\section{Conclusion}
We have shown \textsf{BMLF} to be \NP-hard even if 
every fence contains at most three polygons,
each fence has at most five links, 
and only two different colors of polygons are present.
Our reduction holds regardless of requiring disjoint fences or not.
Note, that our reduction can be adapted to not require the outer bounding polygon $Q$.
Instead, we can replace $Q$ by one polygon with a narrow and very complex channel,
connecting the ``inside'' with the ``outside''.
On the algorithmic side,
we gave an \XP-algorithm for \textsf{BMLF} 
parameterized by the maximum number of links in a fence and
allowing at most two polygons per fence.
We also showed that \isaacrev{two-colored} \textsf{CMLF} can be solved in polynomial time.

It is open if one can eliminate the exponential dependency on
the number of links in our algorithm for \textsf{BMLF}.
Furthermore, while our reduction holds when replacing the outer bounding polygon,
our algorithm does not since we cannot immediately apply \cref{thm:klinkblackbox}.
Similarly, requiring the fences to be disjoint for \textsf{BMLF} is an interesting open
direction.

\bibliography{fencing_refs}

\clearpage

\ifArxiv\else \end{document}\fi

\appendix

\section{Omitted proofs of Section~\ref{sec:nphard}}\label{apx:nphard}
\lowerbound*
\begin{proof}
    By Observation~\ref{obs:kbends} the fence $F$ has at least $t-1$ straight-line segments 
    Assume $F$ has also exactly $t-1$ segments.
    Let $\mathcal S_1,\ldots,\mathcal S_t$ be the non-overlapping non-collinear triples and
    $S_1,\ldots,S_{3t}$ be the sets of the non-collinear triple in the order in which $F$ crosses them as above.
    Since $F$ is a simple polygon there exists a polygonal chain $C$ starting at some point in $S_{3t}$ and
    ending at some point in $S_1$.
    By Observation~\ref{obs:kbends} there are $t-1$ straight-line segments 
    which are completely contained in $]S_1,S_{3t}[$ which implies
    that $C$ must consist of segments that are also completely contained in $]S_1,S_{3t}[$.

    Now, we charge every of the $t-1$ segments to the piece of $F$ that 
    connects for each $\mathcal S_i$ with $i = 1,\ldots,t-1$
    some point in $S_{i+1}$ with some point in $S_{i+2}$.
    Since the triples are non-overlapping each such connection also requires a distinct segment.
    Then, for the connection of $S_{3t-1}$ and $S_{3t}$ we require at least one more segment.
    Let $s$ be the segment charged to the connection of $S_{3t-4}$ to $S_{3t-3}$.
    If we could extend this segment to also intersect $S_{3t}$ we would violate
    that $\mathcal S_{t-1}$ and $\mathcal S_t$ are non-overlapping.
    Symmetrically for the segment charged to $S_2$ and $S_3$.
    Consequently, we require at least one more segment.
\end{proof}

\subsection{Omitted proofs of Section~\ref{sec:variable}}\label{apx:nphard:variable}
\vargadget*
\begin{proof}
    To prove this lemma, we will first number the eight inner polygons of the variable gadget $A_1, \ldots, A_8$.
    Any fence inside the variable gadget can include any combination of $k$ inner polygons.
    We will prove this lemma, by enumerating for every $k \in \{1, 8\}$ all possible (non-symmetric) combinations of including these inner polygons; Note that technically a fence is allowed to include non-consecutive inner polygons.
    Note that any fence has to consist of at least 3 segments and clearly every inner polygon of $\mathcal{G}(v)$ can be fenced alone with 3 segments.
    For any possible fence including multiple inner polygons, we show the existence of a certain number of non-overlapping non-collinear triples, which provide a lower bound on the number of segments for such a fence.
    All triples and the resulting lower bounds on the number of segments in fencings including all possible combinations of 2 to 8 inner polygons are shown in Appendix~\ref{sec:app:variable_gadget}.
    Enumeration shows that a fence $F$ including $k$ inner polygons, has the minimum amount of segments, if and only if the indices of the contained inner polygons are consecutive (assuming $A_8$ and $A_1$ to be consecutive).
    In particular, such a fence including 2, 3, 4, 5, 6, 7 or 8 inner polygons, requires (and can be realized with) 3, 6, 8, 11, 13, 16 and 18 segments, respectively.

    First observe that there are exactly two minimum link fencings, which include exactly two neighboring inner polygons in one fence, both consisting of 12 segments total, both are shown in Figures~\ref{fig:variable_gadget_true} and~\ref{fig:variable_gadget_false}.
    We can therefore exclude any fence including six or more inner polygons, since they clearly require more segments by themselves, which immediately eliminates the possibility of including 6 or more inner polygons in one fence.
    Any fence including five inner polygons requires at least 11 segments and at least one additional fence is needed, which increases the segment count to at least 14.
    Any fence including four inner polygons requires at least 8 segments. If the remaining four segments are fenced together, we require at least 16 segments. If they include a group of three polygons, we require at least 14 segments. If two remaining polygons are grouped, we require 11 segments, but at least one more fence (14 segments in total) are needed and if all remaining polygons are fenced alone, we need 20 segments in total.

    Next if three segments are fenced together, we require 6 segments and have five inner polygons left. We already know that no fence including four or more segments can be part of a minimum link fencing.
    If three of the remaining five polygons are grouped we arrive at a total of at least 12 segments, with at least one more fence needed.
    If at most two polygons are grouped, we need at least three more fences and arrive at a total of at least 15 segments.
    
    Finally, polygons could be fenced individually.
    Clearly there must be an even number of individually fenced polygons.
    If all eight polygons are fenced alone, we reach at least 24 segments, for six inner polygons fenced alone we get 18 segments, four individually fenced polygons lead to 12 with at least one more fence needed.
    Two inner polygons being fenced alone require at least 6 segments, with at least three more fences needed for the six left over polygons, which require at least 9 segments leading to a total of 15 segments.
    
    Therefore only two minimum link fencings exist, which require exactly 12 segments, and they group two neighboring inner polygons pairwise.
    We call the fencing, which groups inner polygons, which are both contained in a gray triangle with a true spike $\mathcal{F}_t$ and the other $\mathcal{F}_f$.
\end{proof}

\subsection{Omitted proofs of Section~\ref{sec:clause}}\label{apx:nphard:clause}

\paragraph*{Omitted lemmas and proofs for one isolated wire}
\begin{figure}
    \centering
    \includegraphics[page=20]{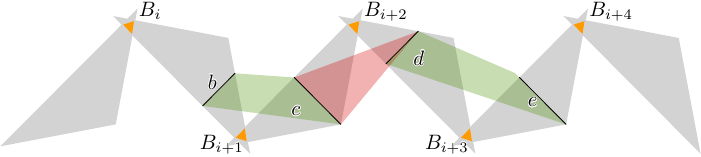}
    \caption{The highlighted areas show that $b$ and $c$, $c$ and $d$, and $d$ and $e$ form non-collinear triples.}
    \label{fig:no_triple_bypass}
\end{figure}
\begin{restatable}[$\star$]{lemma}{consecutivebipass}\label{lem:consecutive_bypass}
    No fence in a minimum link fencing $\mathcal F$ of an isolated wire $W$ of $\mathcal G(c)$
    bypasses two or more consecutive inner polygons of $W$.
\end{restatable}

\begin{proof}
    Let $B_1,\ldots,B_k$ be the inner polygons of an isolated wire $W$,
    $T_1,\ldots,T_k$ the corresponding gray triangles, and
    $F \in \mathcal F$ a fence that bypasses two or more consecutive inner polygons of $W$.
    Also, throughout the proof, let $B_i$ and $B_j$ with $i < j$ be two inner polygons 
    that are contained in $F$ such that $B_{i+1},\ldots,B_{j-1}$ are bypassed by $F$.

    Assume that $i + 1 < j - 1$, i.e., there are at least three consecutive inner polygons bypassed by $F$.
    Let $B_{i+1}, B_{i+2}$ and $B_{i+3}$ be three such consecutive polygons.
    Now, we find one sequence of at least three non-overlapping non-collinear triples by construction.
    Let $b$, $c$, $d$, and $e$ be four sets such that $b$ and $c$, $c$ and $d$, and $d$ and $e$ form such triples.
    Compare Figure~\ref{fig:no_triple_bypass} for an illustration.

    By Observation~\ref{obs:kbends} $F$ contains at least two complete segments in $]b,e[$.
    Moreover, $F$ has to cross the same triples in reverse order since it is a simple polygons and 
    contains both $B_i$ and $B_j$ for $j > i+3$.
    Consequently, we find another two complete segments in $]e,b[$.
    We split $F$ at $b$ and $e$.
    This creates three fences $F_1$, $F_2$, and $F_3$ for which it holds that
    $|F_1| + |F_2| + |F_3| = |F| + 4$
    One fence does not contain any inner polygons and can be deleted,
    let $F_2$ be this fence.
    Sine $F_2$ was created by splitting along $b$ and $e$ it contained at least four complete segments plus
    the two segments introduced in the splitting operation.
    Hence, with $|F| + 4 = |F_1| + |F_2| + |F_3| \geq |F_1| + |F_3| + 6$ it follows that
    $|F| \geq |F_1| + |F_3| + 2 > |F_1| + |F_3|$.
    Since $F$ was part of a minimum link fencing we may conclude that $F$
    at most bypasses two consecutive inner polygons.

    \begin{figure}
        \centering
        \includegraphics[page=7]{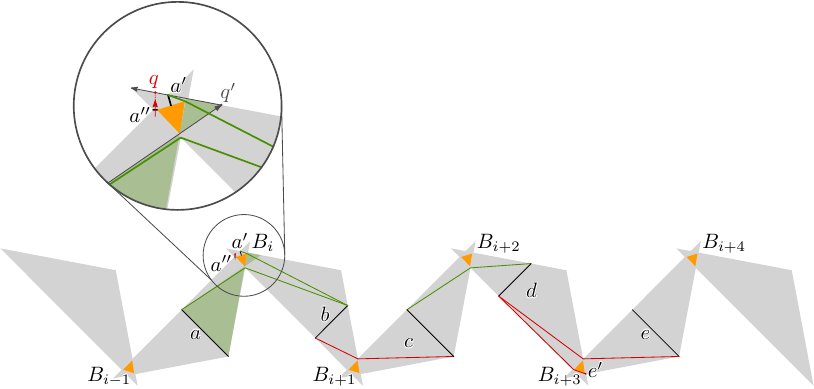}
        \caption{We can locally shortcut any fence $F$ passing $a$ and $b$ in both directions by connecting the two intersection points of $F$ with $a$ ($b$)}
        \label{fig:no_double_bypass}
    \end{figure}
    
    In the following let $B_{i+1}$ and $B_{i+2}$ be the two consecutive polygons bypassed by $F$.
    Consequently, we know that $F$ contains $B_i$ and $B_{i+3}$.
    Figure~\ref{fig:no_double_bypass} illustrates the following cases.
    
    Assume $F$ does not contain any other inner polygons, i.e., $F$ contains only two inner polygons.
    In this case we find a sequence of four non-overlapping non-collinear triples 
    $a'$ and $b$, $b$ and $c$, $c$ and $d$, $d$ and $e'$ which 
    have to be crossed by $F$.
    Moreover, $F$ has to cross each such triple in the reverse direction.
    Consequently, $F$ has crossed eight non-overlapping non-collinear triples and 
    consists by Lemma~\ref{lem:lowerbound} of at least eight segments.
    Replacing $F$ by two fences, both consisting only of triangles, 
    creates a fencing with less segments than $\mathcal F$.
    It remains to argue the case that $F$ contains more than two inner polygons.
    
    Let $F$ be such that there does not exist an inner polygon $B_z$ contained in $F$ with $z > i+3$.
    Then there exists an inner polygon $B_m \in F$ and $m<i$.
    Observe that there is a sequence of four non-overlapping non-collinear triples 
    $a'$ and $b$, $b$ and $c$, $c$ and $d$, $d$ and $e'$, which is as in the case before and
    a sequence of four non-overlapping non-collinear triples
    $e'$ and $d$, $d$ and $c$, $c$ and $b$, and $b$ and $a$ 
    which is almost as before, with the difference that we use a set $a$ which lies in $T_{i-1}$ and
    can be chosen as a segment which $F$ intersects twice.
    By Observation~\ref{obs:kbends} we find that $]a',a[$ 
    completely contains at least seven segments of $F$.
    
    We split $F$ at $a$, let $F_m$ and $F_i$ be the resulting fences such that
    $F_m$ contains $B_m$ and $F_i$ contains $B_i$ and $B_{i+3}$.
    It holds that $|F_m| + |F_i| = |F| + 2$ and $|F_i| \geq 8$ since $F_i$ contains only $B_i$ and $B_{i+3}$ and
    bypasses $B_{i+1}$ and $B_{i+2}$.
    Hence, $|F| \geq |F_m| + 6$ since $|F| + 2 = |F_m| + |F_i| \geq |F_m| + 8$.
    Consequently, replacing $F_i$ by two triangles only yields a fencing with 
    the same number of fences as $\mathcal F$.
    In the following assume that we delete $F_i$ and 
    introduce one triangular fence $F_{\Delta}$ that includes only $B_{i+3}$.
    The goal is then to include $B_i$ into $F_m$ using at most three segments.
    
    Consider $F_m$ and remove the segment introduced when splitting $F$ at $a$.
    Let $p$ and $p'$ be the two intersection points of $F$ with $a$ and
    assume $p$ is closer to the base side of $T_{i-1}$.
    W.l.o.g. we assume that neither $p$ nor $p'$ are vertices of $F$.
    Let $c_i$ be the corner of $T_{i-1}$ that would be disconnected from the component containing $a$
    when removing $T_i$ from the plane.
    If there exists a point $c$ in the component containing $c_i$ such that 
    the straight-line segments $pc$ and $p'c$ both do not intersect $B_i$,
    are completely contained in $T_{i-1}$, and
    their supporting lines leave $B_i$ in different half-planes,
    we just add these two segments to $F_m$ which now contains also $B_i$.
    
    Now assume such a point does not exist.
    By construction there exists a straight-line segment starting at $p$, $p'$ respectively,
    ends at the boundary of $T_i$, and
    also intersects the boundary of $T_i$.
    Let $s$ be such a segment for $p$ and $s'$ one for $p'$ and 
    let $q$ and $q'$ be their endpoints on the boundary of $T_i$.
    Sine $p$ is closer to the base of $T_{i-1}$ we may assume that 
    the supporting line of $s$ leaves $B_i$ to the right and 
    the one of $s'$ leaves $B_i$ to the left, and
    $s$ and $s'$ do not intersect.
    Connect the two endpoints of $s$ and $s'$ with one segment along the boundary of $T_i$.
    This uses at most three segments as required.
    
    In preparation for the next case we are going to remove one additional segment.
    Consider a set $a''$ that is intersected by $s$, such a set is indicated in Figure~\ref{fig:no_double_bypass}.
    Since $F$ contained $B_i$ it had to contain at least one point of $a''$ as well.
    Moreover, since $F$ was a minimum link fence containing $B_i$%
    \footnote{We are only aiming to contradict minimality with respect to inclusion of polygons in the fence,
    for a given set of polygons we may still assume that the initial fence was as short as possible.} 
    we may assume that there is a vertex $v$ of $F$ that is now a vertex of $F_m$ and
    there is a straight-line segment starting at $v$,
    intersecting $a''$,
    intersecting the boundary of $T_i$ twice, and 
    it lies in $T_{i-1}$ and does not intersect $B_i$.
    Let $q''$ be the endpoint of this segment.
    Replace the segments $vp$, $pq$, and $qq'$ by $vq''$ and $qq''$.
    That is only two additional segments.
    
    Finally, assume that there exist inner polygons $B_m,B_l\in F$ such that
    $m < i$ and $l > i + 3$.
    Split $F$ as above at $a$ and also at some set $e$ in $T_{i+3}$.
    This leaves three fences $F_m$, $F_l$, and $F_i$.
    Similarly to above $|F| + 4 = |F_m| + |F_l| + |F_i| \geq |F_m| + |F_l| + 8$ from which we derive
    $|F| \geq |F_m| + |F_l| + 4$.
    By the above argumentation we can hence replace $F$ by two fences $F_m$ and $F_l$ such that
    the new fencing has the same number of segments.
    But now, observe that there is a sequence of non-overlapping non-collinear triples
    which $F_i$ all has to cross, namely
    $a'$ and $b$, $b$ and $c$, $c$ and $d$, $d$ and $e$, $e'$ and $d$, $d$ and $c$, $c$ and $b$, and $b$ and $a$.
    By Observation~\ref{obs:kbends} intervals $]a',e[$ and $]e', a[$ completely contain three segments each.
    Moreover, these sets can be chosen in such a way that the segments introduced when splitting $F$ 
    do only contain points of $a$ and $e$.
    See also Figure~\ref{fig:no_double_bypass} for an illustration.
    Hence, $|F_i| \geq 10$ and consequently $|F| \geq |F_m| + |F_l| + 6$.
  \end{proof}  
    
We now know that no fence of a minimum link fencing of an isolated wire bypasses two or more consecutive inner polygons of that wire.
However, it might still bypass an unbounded number of polygons in total.
Since there is a sequence of four non-overlapping non-collinear triples for a fence that 
includes $B_i$ and $B_{i+2}$ we obtain the following observation for a fence bypassing at least one inner polygon.

\begin{observation}\label{obs:onetriples}
    Let $\mathcal F$ be a minimum link fencing of an isolated wire $W$ of $\mathcal G(c)$,
    let $B_1,\ldots,B_k$ be the inner polygons of $W$, and
    $F\in\mathcal F$ a fence that bypasses $z > 0$ inner polygons,
    then $F$ contains at least $4z + 2$ segments.
\end{observation}

While it is not possible anymore by Lemma~\ref{lem:consecutive_bypass} to bypass multiple consecutive polygons
there could still be multiple fences in a minimum link fencing that interleave and
bypass many individual inner polygons.

\begin{restatable}[$\star$]{lemma}{nointerleave}\label{lem:no_interleave}
    A minimum link fencing $\mathcal F$ of an isolated wire $W$ of $\mathcal G(c)$
    does not contain two distinct fences that interleave.
\end{restatable}
\begin{proof}
    Let $B_1,\ldots,B_k$ be the inner polygons of $W$ and 
    assume that there are fences $F,F'\in \mathcal F$ of $W$ that interleave.
    By Lemma~\ref{lem:consecutive_bypass} we know neither $F$ nor $F'$ 
    can bypass two or more consecutive inner polygons.
    Consequently, there exist inner polygons $B_i$, $B_{i+1}$, $B_{i+2}$, and $B_{i+3}$ such that 
    w.l.o.g.~$B_i, B_{i+2} \in F$ and $B_{i+1}, B_{i+3} \in F'$.
    
    First, assume that $F$ and $F'$ only contain these four polygons.
    By Observation~\ref{obs:onetriples} each fence bypassing one polygon contains at least six segments.
    Hence, $|F| + |F'| \geq 12$. 
    By replacing $F$ and $F'$ by two new fences, 
    consisting of a triangle each, one containing $B_i$ and $B_{i+1}$ and 
    the other containing $B_{i+2}$ and $B_{i+3}$, we create a new fencing with six segments less.
    A contradiction to $\mathcal F$ being a minimum link fencing.
    
    Now consider the case that $F$ and $F'$ might contain more inner polygons.
    First assume all further inner polygons of $F$ are before $B_i$ and 
    all further inner polygons of $F'$ are after $B_{i+3}$.
    Observe that then $|F| + |F'| \geq 14$ since we need to add at least one segment to each fence.
    We split the fences $F$ and $F'$ in the gray triangles $T_{i-1}$, $T_{i+3}$ respectively.
    We first consider just $F$.
    After splitting we now have two fences $F_{<i}$ and $F_i$.
    Let $F_{<i}$ be the fence containing all inner polygons before $B_i$ and
    $F_i$ the fence containing $B_i$ and $B_{i+2}$.
    By Observation~\ref{obs:onetriples} we have $|F_i| \geq 6$.
    Consequently $|F| = |F_{<i}| + |F_i| - 2 \geq |F_{<i}| + 6 - 2 = |F_{<i}| + 4$ and hence
    $|F_{<i}| \leq |F| - 4$.
    Doing the same for $F'$ we obtain $F_{<i}'$ and $F_i'$.
    As before we can replace $F_i$ and $F_i'$ by two triangles $F_\Delta$ and $F_\Delta'$.
    In total the new fencing created like this has 
    $|F_{<i}| + |F_\Delta| + |F_{<i}'| + |F_\Delta'| \leq |F| + |F'| - 8 + 6 = |F| + |F'| - 2$
    segments.
    
    Finally, $F$ and $F'$ could contain polygons both before $B_i$ and after $B_{i+3}$.
    By Lemma~\ref{lem:consecutive_bypass} we know that then the fences repeatedly interleave until
    at some inner polygon one fence stops.
    Let $m < l$ and w.l.o.g. 
    let $B_m$ be the last polygon included by $F$ and $B_l$ the first polygon included by $F'$ and
    assume that between $B_l$ and $B_m$ the two fences interleave, i.e.,
    $F$ includes $B_m$, $B_{m-2}$, etc. and $F'$ contains 
    $B_{m-1}$, $B_{m-3}$
    , etc..
    Moreover, $F$ contains $B_{l-1}$ and $F'$ contains $B_{m+1}$
    We can split $F$ and $F'$ as before in $T_{l-2}$ and $T_{m+2}$ using just four segments.
    Let $F_{<m}$ and $F_m$ be the fences created from splitting $F$ such that
    $F_m$ contains $B_m,\ldots,B_{l-1}$ and $F_{<m}$ the remaining polygons contained in $F$.
    Analogously for $F'$ and $F_{<l}$ and $F_l$.
    Let $M$ be the number of bypassed polygons for $F$ and 
    $L$ the number of bypassed polygons for $F'$.
    We get that $|F| + |F'| = |F_{<m}| + |F_m| + |F_{<l}| + |F_l| - 4$.
    We also know from Observation~\ref{obs:onetriples} that $|F_m| \geq 4M + 2$ and $|F_l| \geq 4L + 2$.
    Hence, $|F| + |F'| \geq |F_{<m}| + |F_{<l}| + 4M + 4L$ and hence $\soerenchange{|}F_{<m}\soerenchange{|} + \soerenchange{|}F_{<l}\soerenchange{|} \leq |F| + |F'| - 4M - 4L$.
    Finally, we can replace $F_m$ and $F_l$ by a series of triangular fences that in sum have $3 \cdot (M+L+2)/2$ which
    results in 
    \begin{align*}
        |F_{<m}| + |F_{<l}| + 3/2(M+L+2) - 4 
        &\leq |F| + |F'| - 4M - 4L + 3/2(M+L+2) - 4\\
        &=    |F| + |F'| - (5/2(M+L) + 1).
    \end{align*}
    This concludes the proof as $\mathcal F$ cannot have been a minimum link fencing for $W$.
 \end{proof}
 
 \begin{figure}
    \centering
    \begin{subfigure}[b]{.45\linewidth}
        \centering
        \includegraphics[page=8]{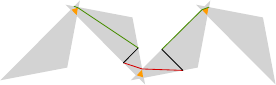}
    \end{subfigure}
    \qquad
    \begin{subfigure}[b]{.45\linewidth}
        \centering
        \includegraphics[page=9]{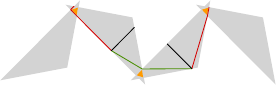}
    \end{subfigure}
    \caption{All six non-collinear triples of a fence bypassing exactly one inner polygon.}
    \label{fig:no_interleave}
\end{figure}
 
\consecutive*
\begin{proof}
    Let $B_1,\ldots,B_k$ be the inner polygons of $W.$
    Assume that $F$ bypasses an inner polygon $B_j$ for $2\leq j\leq k-1$.
    By Lemma~\ref{lem:consecutive_bypass} we know that $F$ never bypasses more than 
    one consecutive inner polygon at a time.
    Hence, we know that $B_{j-1}$ and $B_{j+1}$ are both contained in $F$.
    We distinguish if $F$ bypasses $B_j$ above or below (both cases are shown in Figure~\ref{fig:no_bypass}).
    
    Assume $F$ bypasses $B_j$ above as shown in Figure~\ref{fig:no_bypass}~(a).
    Then we construct a new fence including also $B_j$ as follows.
    Let $s_1$ and $s_2$ be two segments such that
    $s_1$ lies inside the gray triangle $t_{i-1}$ of $W$ that contains $B_{j-1}$ and $B_j$ and
    $s_2$ such that it lies inside the gray triangle $t_{i}$ of $W$ containing $B_j$ and $B_{j+1}$.
    More specifically, we choose $s_1$ such that its supporting line leaves $B_{j-1}$ and $B_j$ in one and 
    $B_{j-2}$ in the other half-plane.
    Similarly, we chose $s_2$ such that its supporting line leaves $B_j$ and $B_{j+1}$ in one and
    $B_{j+2}$ in the other half-plane.
    The segments $s_1$ and $s_2$ can then be extended such that they meet in a point that is inside
    the overlap of $t_{i-1}$ and $t_{i}$ and below $B_j$.
    Moreover, they can be extended such that they intersect any fence that contains
    $B_{j-1}$ and $B_{j+1}$ at least twice below $B_{j-1}$ and $B_{j+1}$ respectively.
    See Figure~\ref{fig:no_bypass}(a) for an illustration.
    
    Symmetrically we find two segments $s_1'$ and $s_2'$ whose intersection point is above $B_j$ and
    whose supporting lines leave $B_{j-1}$ and $B_j$ in the same half-plane and $B_{j+1}$ in the other,
    $B_{j+1}$ and $B_j$ in the same and $B_{j-1}$ in the other respectively.
    Again, these segments can be extended to intersect $F$ twice, this time above $B_{j-1}$ and $B_{j+1}$.
    See Figure~\ref{fig:no_bypass}(b) for an illustration.
    
    As $F$ contains at least four segments inside $t_{i-1}$ and $t_i$ we can replace those by
    $s_1$, $s_2$, $s_1'$, and $s_2'$ which yields a fence with at most equal number of links.
    Now, deleting the single fence that fenced $B_j$ removes at least three links, a contradiction to
    $\mathcal F$ being a minimum link fencing.

    \begin{figure}[tbp]
        \centering
        \begin{subfigure}[b]{.45\linewidth}
            \centering
            \includegraphics[page=2]{figures/no_bypass.pdf}
        \end{subfigure}
        \qquad
        \begin{subfigure}[b]{.45\linewidth}
            \centering
            \includegraphics[page=1]{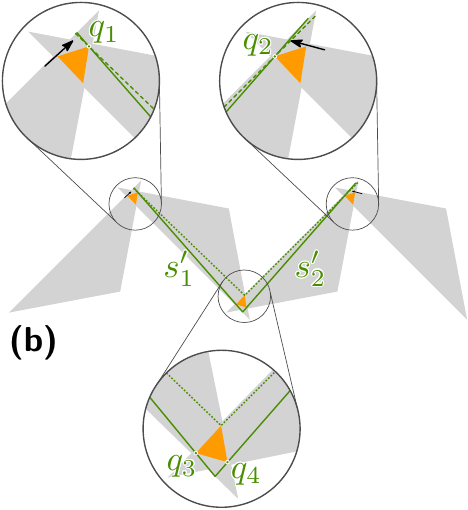}
        \end{subfigure}
        \caption{Rerouting of segments to include a bypassed inner polygon.}
        \label{fig:no_bypass}
    \end{figure}
    
    We can apply this procedure to all bypassed polygons.
    Since no fence can be interleaving with $F$, we can simply remove all fences, which included any polygon $B_i$ completely and still obtain a valid fencing, contradicting
    that $\mathcal F$ is a minimum link fencing.
\end{proof}

\onlytwofences*
\begin{proof}
    Let $B_1,\ldots,B_k$ be the inner polygons of $W$.
    By Lemma~\ref{lem:sizebound} every fence in $\mathcal F$ contains 
    either one, two, or three consecutive inner polygons.
    Let $f_i$ be the number of fences in $\mathcal F$ containing $i$ inner polygons.
    Then the number of segments of the fencing $\mathcal F$ can be computed as 
    $|\mathcal F| \geq 3f_1 + 3f_2 + 6f_3$ since for one or two inner polygons we always can use one triangle 
    and for three inner polygons we have six segments by Lemma~\ref{lem:conspolys}.
    Now, let $k_i$ be the number of segments
    in a fence containing $i$ inner polygons.
    From this we get that $f_i = k_i/i$.
    Substituting in the previous calculation we get $|\mathcal F| \geq 3k_1 + 3/2k_2 + 2k_3$.
    Hence, maximizing the number of inner polygons in fences containing also two polygons 
    minimizes the number of segments in the fencing.
    Since $k$ is always even the minimum is attained at $3/2k$, i.e., when all inner polygons are part of a fence 
    containing only two inner polygons as claimed.
\end{proof}

\paragraph*{Omitted proofs for integrating the clause triangle}    

\nobypassinginclause*
\begin{proof}
    Assume there exists a fence $F$ that bypasses some polygons of $\mathcal G(c)$.
    If $F$ only bypasses and contains polygons of one isolated wire $W$ and
    is contained in the gray triangles of $W$ we can apply Lemma~\ref{lem:consecutive}
    contradicting the existence of $F$.
    
    Next, assume that $F$ lies not only in the gray triangles of $W$ but still 
    only contains and bypasses inner polygons of $W$.
    This implies that $F$ contains $B_1$ of $W$ as else we could by construction find a fence that
    is restricted to only the gray triangles of $W$.
    Let $T_c$ be the triangle intersecting $T_1$ of $W$ and containing $B_c$.
    If $F$ was not restricted to the $T = T_c \cup T_1 \cup \ldots \cup T_k$ we could
    replace $F$ by a congruent minimum link fence that is only contained $T$ by
    either splitting $F$ if a whole segment is outside of $T$ or
    moving the one corner that lies outside of $T$ which is always possible by construction.
    After that we can again apply Lemma~\ref{lem:consecutive} 
    by considering $T_c$ and $B_c$ part of the isolated wire.
    
    This means that $F$ has to include either $B_c$ or two inner polygons of different wires if
    $F$ is to bypass any polygon of $\mathcal G(c)$.
    Now assume $F$ was bypassing any inner polygon $B_i$ of some wire $W$ and $i > 1$.
    Let $F_T$ be $F$'s restriction to $T = T_c \cup T_1 \cup \ldots \cup T_k$.
    Add one temporary segment to $F_T$ to close the fence, this can always be done in $T_c$.
    Now we can apply Lemma~\ref{lem:consecutive} to $F_T$.
    Each proof of one of the lemmas leading to Lemma~\ref{lem:consecutive}
    implies a procedure of how to split $F_T$.
    This results in two or more fences.
    Among those exists one that contains $B_j$ with smallest $j$.
    Let $F_T^j$ be this fence. 
    Either $F_T^j$ is such that it still contains the segment added initially to $F_T$
    in which case we just remove this segment and obtained a fencing with less segments or
    $F_T^j$ does not contain this segment anymore.
    Recall, that every procedure we implicitly defined saves at least two segments.
    Hence, adding the segment we had added to $F_T$ to the remainder of $F - F_T$ 
    still creates a fencing with one segment less.
    
    Consequently, we may assume that $F$ bypasses any combination of $B_1^1$, $B_1^2$, $B_1^3$, or $B_c$.
    Without loss of generality we assume a fence $F$ bypasses 
    the first inner polygon $B_1^1$ of wire $W_1$.
    We can again show that the amount of non-collinear triples
    between the lines $a$ and $b$ is four, inducing at least four bends in this part of $F$.

    \begin{figure}
        \centering
        \includegraphics[page=14]{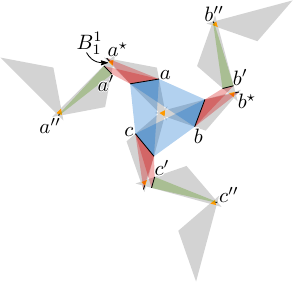}
        \caption{The twelve sets.}
        \label{fig:twelvesets}
    \end{figure}
    
    We define twelve sets $a, a^\star, a', a'', b, b^\star, b', b'', c, c^\star, c'$ and $c''$, as shown in Figure~\ref{fig:twelvesets}.%
    Any fence crossing the following pairs of sets contains a non-collinear 
    triple starting at one set and ending at the other:
    \begin{itemize}
        \item $a$ and $b$
        \item $a$ and $c$
        \item $b$ and $c$
        \item $a$ and $a^\star$
        \item $a$ and $a'$
        \item $a'$ and $a''$
        \item $b$ and $b^\star$
        \item $b$ and $b'$
        \item $b'$ and $b''$
        \item $c$ and $c^\star$
        \item $c$ and $c'$
        \item $c'$ and $c''$
    \end{itemize}

    Recall that $F$ is a fence bypassing $B_1^1$ hence $F$ has to include some $B_j^1$ for $j > 1$.
    More precisely, we may assume that $F$ contains $B_2^1$ as else we could find a shorter fence
    by the above discussion.
    Let $F'$ be the fence that contains $B_1^1$.
    We make a case distinction over which subset of $\{B_c,B_1^2,B_2^1,B_1^3,B_2^3\}$ the fence $F'$ contains as well.

    Note that regardless of inclusion or exclusion of the clause triangle, any fence crossing $a$ and $b$, $a$ and $c$ or $b$ and $c$ has one non-collinear triple starting at one and ending at the other line segment.
    We therefore only analyze the cases, which do not contain the clause triangle.

    We analyze these cases one by one. 
    In all cases we will turn $F$ and $F'$ into shorter fences $\hat F$, $\hat F_1'$, and $\hat F_2'$ where
    $\hat F$ is going to be the fence including at least $B_1^1$ and $B_1^2$ and
    is formed by just including $B_1^1$ instead of bypassing it.
    The fences $\hat F_1'$ and $\hat F_2'$ containing the remaining triangles 
    are going to be formed by splitting $F'$ at some point and removing the remaining empty part.
    We call this \emph{shortcutting} $F'$ at the segment where we split it.
    Note that we can include $B_1^1$ into $\hat F$ at no additional cost.
    This reduces the number of segments in $\mathcal F$ contradicting
    that it is a minimum link fencing.
    To shorten the notation we write $[a,b,c]$ for a sequence of non-collinear triple that $F$ has to cross.

    \subparagraph{$\emptyset$: }
    If $B_1^1$ is contained in its own fence, we could instead reroute $F$ 
    (Figure~\ref{fig:b_1_bypass_triples}), 
    and omit $F'$ entirely, reducing the number of segments.
    
    \subparagraph{$\{B_1^2\}$ and $F'$ has to cross the sets $[a^\star, a, b, b^\star, b, a, a^\star]$: }
    We can shortcut $F$ along $b$ for an additional cost of one segment, 
    which omits three segments, 
    which are necessarily contained in the sequence $[b, a, a^\star, a, b]$.
    
    \subparagraph{$\{B_2^2\}$ and $F'$ has to cross the sets $[a^\star, a, b, b', b'', b', b, a, a^\star]$: }
    This fence can be shortcut along $b'$ for an additional cost of one segment, 
    which omits five segments, 
    which are necessarily contained in the sequence $b', b, a, a^\star, a, b, b'$.
    
    \subparagraph{$\{B_1^2,B_2^2\}$ and $F'$ has to cross the sets $[a^\star, a, b, b', b'', b', b, a, a^\star]$: }
    This case can be resolved in the same way as $\{B_1^2\}$.
    
    \subparagraph{$\{B_2^2,B_2^3\}$ and $F'$ has to cross the sets $[a^\star, a, b, b', b'', b', b, c, c', c'', c', c, a, a^\star]$: }
    This fence can be shortcut along $b'$ and along $c'$, 
    splitting it into two fences containing only inner polygons of the same wire.
    This can be done at an additional cost of two segments saving five segments, 
    which are necessarily contained in the sequence $b', b, a, a^\star, a, c, c'$.
    
    \subparagraph{$\{B_1^2,B_2^2,B_2^3\}$ and $F'$ has to cross the sets $[a^\star, a, b, b', b'', b', b, c, c', c'', c', c, a, a^\star]$: }
    This fence can be shortcut along $b$ and along $c'$, 
    splitting it into two fences containing only inner polygons of the same wire.
    This can be done at an additional cost of two segments saving four segments, 
    which are necessarily contained in the sequence $b, a, a^\star, a, c, c'$ and one segment, 
    necessarily contained in the sequence $c', c, b$.
    
    \subparagraph{$\{B_1^2,B_2^2,B_1^3,B_2^3\}$ and $F'$ has to cross the sets $[a^\star, a, b, b', b'', b', b, c, c', c'', c', c, a, a^\star]$: }
    This fence can be shortcut along $b$ and along $c$, 
    splitting it into two fences containing only inner polygons of the same wire.
    This can be done at an additional cost of two segments saving three segments, 
    which are necessarily contained in the sequence $b, a, a^\star, a, c$.
    
    \subparagraph{$\{B_1^2,B_1^3\}$ and $F'$ has to cross the sets $[a^\star, a, b, b^\star, b, c, c^\star, c, a, a^\star]$: }
    This case can be resolved in the same way as $\{B_1^2,B_2^2,B_1^3,B_2^3\}$.
    
    \subparagraph{$\{B_1^2,B_2^3\}$ and $F'$ has to cross the sets $[a^\star, a, b, b^\star, b, c, c', c'', c', c, a, a^\star]$: }
    This case can be resolved in the same way as $\{B_1^2,B_2^2,B_2^3\}$.
    
    \subparagraph{$\{B_1^2,B_2^2,B_1^3\}$ and $F'$ has to cross the sets $[a^\star, a, b, b', b'', b', b, c, c^\star, c, a, a^\star]$: }
    This case can be resolved in the same way as $\{B_1^2,B_2^2,B_1^3,B_2^3\}$.
 \end{proof}   
 
\begin{figure}
    \centering
    \begin{subfigure}[b]{.48\linewidth}
        \centering
        \includegraphics[page=10]{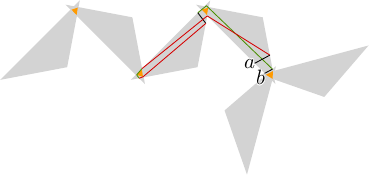}
    \end{subfigure}
    \quad
    \begin{subfigure}[b]{.48\linewidth}
        \centering
        \includegraphics[page=11]{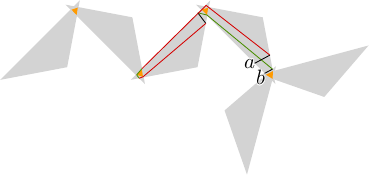}
    \end{subfigure}
    \caption{Non-collinear triples for any fence bypassing $B_1$}
    \label{fig:b_1_bypass_triples}
\end{figure}

\onlyinonewire*
\begin{proof}

    \begin{figure}
        \centering
        \begin{subfigure}[b]{\linewidth}
            \centering
            \includegraphics[page=21]{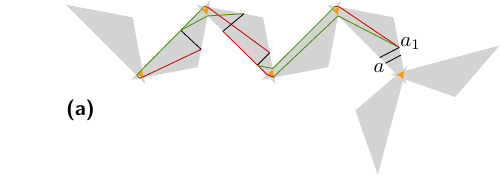}
        \end{subfigure}
        \begin{subfigure}[b]{\linewidth}
            \centering
            \includegraphics[page=22]{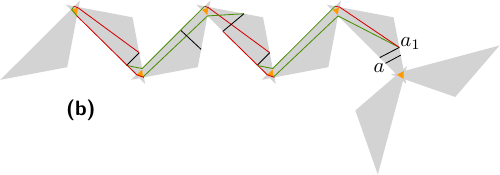}
        \end{subfigure}
        \caption{Non-collinear triples in a fence including at least (a) four or (b) five inner polygons of a wire.}
        \label{fig:non_collinear_for_four}
    \end{figure}
    Assume there is a fence $F$, which includes inner polygons of two wires.
    By Lemma~\ref{lem:no_bypassing_in_clause}, we know that all inner polygons of a wire included in $F$ are consecutive in that wire.
    First also assume that $F$ contains $m>3$ inner polygons of a single wire $W_1$.
    We split such a fence at a set $a$ in $T^1_1$ (again using the sets marked in Figure~\ref{fig:twelvesets}) into two fences $\hat{F}$ containing all $m$ inner polygons of the wire and
    $F_o$ and we have $|\hat{F}| + |F_o| = |F| + 2$.
    By Lemma~\ref{lem:conspolys}, we know that $\hat{F}$ contains at least $2m \geq 8$ segments.
    Assume $m$ to be odd, then we can replace $\hat{F}$ with fences containing two consecutive polygons each (exactly $\frac{m-3}{2}$), plus one containing three inner polygons, resulting in $\frac{3m - 9}{2} + 6$ segments, which is smaller than $2m - 2$ for any odd $m>5$, yielding a better solution and contradicting $\mathcal F$ being a minimum link fence.
    Assume $m$ to be even.
    Then by Lemma~\ref{lem:onlytwofences}, we obtain $\frac{3m}{2}$ segments in total, which is smaller than $2m - 3$ for any even $m > 4$, again yielding a better solution and contradicting $\mathcal F$ being a minimum link fence.
    
    It remains to analyze the cases $m=4$ and $m=5$.
    Assume $m=4$.
    We obtained $F_o$ by shortcutting $F$ at $a$ and therefore $|F_o| = |F| + 1$.
    Since there are seven segments completely contained in the part of $F$ (the part of $F$ starting and ending at $a$, and being completely contained in $W_1$) which is now omitted
    (using Observation~\ref{obs:kbends}), we can cover the four inner polygons of $W_1$ with two triangles, using only six segments, which contradicts $\mathcal F$ being a minimum link fence.
    
    Assume now $m=5$.
    Since there are seven segments completely contained in the part of $F$ which is now omitted
    (using Observation~\ref{obs:kbends}), we can cover the five inner polygons of $W_1$ with two fences including three and two inner polygons, respectively, i.e., $|\hat{F}_1| + |\hat{F}_2| = 9$.
    Therefore we can transform $F$ into a fence, which includes less than 4 inner polygons of $W_1$ (in this case we turned it into $F_o$, which contains none),
    and it suffices to analyze all cases, in which we contain at most the first one, two or three inner polygons of any wire.
    
    We again enumerate all cases, which are not symmetric to each other.
    We label these cases with the included polygon of each wire with the highest index (since no polygon of a wire is bypassed, this completely characterizes the included inner polygons).
    Also, since we are only investigating non-symmetric cases, we assume that the largest index of a polygon included in the first wire is greater or equal to the one in the second one, which in turn is greater or equal to the one in the third.
    
    The possible cases are given below. For every case, we can give a tight lower bound on the number of segments, which are at least needed, for any fence including these inner polygons. This bound is achieved by analyzing non-collinear triples of such a fence and by providing a fence which achieves this bound. The complete enumeration is shown in Appendix~\ref{sec:app:clause_triangle}.
    \begin{itemize}
        \item $\{B_1^3,B_2^3,B_3^3\}$
        \item $\{B_1^3,B_2^3,B_3^2\}$
        \item $\{B_1^3,B_2^3,B_3^1\}$
        \item $\{B_1^3,B_2^3\}$
        \item $\{B_1^3,B_2^2,B_3^2\}$
        \item $\{B_1^3,B_2^2,B_3^1\}$
        \item $\{B_1^3,B_2^2\}$
        \item $\{B_1^3,B_2^1,B_3^1\}$
        \item $\{B_1^3,B_2^1\}$
        \item $\{B_1^2,B_2^2,B_3^2\}$
        \item $\{B_1^2,B_2^2,B_3^1\}$
        \item $\{B_1^2,B_2^2\}$
        \item $\{B_1^2,B_2^1,B_3^1\}$
        \item $\{B_1^2,B_2^1\}$
        \item $\{B_1^1,B_2^1,B_3^1\}$
        \item $\{B_1^1,B_2^1\}$
    \end{itemize}
    
    For the first twelve of these fourteen cases, we show how $F$ can be replaced with a set of fences, which in sum contain less segments than $F$ contradicting that $\mathcal{F}$ is a minimum link fence; for details, we refer again to Appendix~\ref{sec:app:clause_triangle}).
    All replacement fences (including and excluding the clause triangle) are shown in Appendix~\ref{sec:app:clause_triangle}.
    The only exceptions are the two cases $\{B_1^1,B_2^1,B_3^1\}$ and $\{B_1^1,B_2^1\}$.
    In both cases we can still replace $F$ as shown in Appendix~\ref{sec:app:clause_triangle}, however this yields a fencing, which uses the same amount of segments as $F$.
    Note that both fencings only contain fences, which are completely contained in $T = T_c \cup T_1 \cup \ldots \cup T_k$ (reusing the notation of the proof of Lemma~\ref{lem:no_bypassing_in_clause}).
    The replacements for both cases contain at least one fence, which includes only the first polygon of a wire.
    We will show that this implies that we can replace all fences in that wire with a different set of fences, which uses a smaller number of segments and thereby again contradicting that $\mathcal F$ is a minimum link fencing.

    Let $F_1$ be a fence in a wire including only the first polygon of that wire.
    Since every wire contains an even number of inner polygons and no fence includes four or more inner polygons by Lemma~\ref{lem:sizebound}, we know that this wire has to contain at least one other fence of size one or three.
    Let $F_2 \not = F_1$ be a fence in the wire, s.t., $B_j$ is the inner polygon with the smallest index contained in $F_2$ and no other polygon $B_k$ with $1<k<j$ is contained in a fence of size one or three.
    Therefore there are an even number of inner polygons between $F_1$ and $F_2$.
    In particular if $F_2$ contains three inner polygons, we can create new fences, s.t., $B_2, B_3, B_4$ are contained in one fence of size three and $B_5, \ldots, B_{k+2}$ are contained in $\frac{k-2}{2}$ fences of size two.
    $F_1$ and $F_2$ require nine segments in total, while we can replace them with two fences including $B_1, B_2$ and $B_3, B_4$, respectively, requiring only six segments.
    If $F_2$ contains one inner polygon, we can create new fences, s.t., $B_2$ is contained in one fence of size one and $B_3, \ldots, B_{k}$ are contained in $\frac{k-2}{2}$ fences of size two.
    $F_1$ and $F_2$ require six segments in total, while we can replace them with one fence including $B_1, B_2$, requiring only three segments.
 \end{proof}   

\paragraph*{Omitted proofs for interaction with the variable gadgets}
\clausecost*
\begin{proof}
    Note that inclusion of the last polygon of a wire into a fence of a variable gadget always incurs an additional cost (of two or three, depending on the state of the gadget).
    The only reason a minimum link fencing would choose to do so is the fact that, such an inclusion reduces the number of inner polygons, which have to be fenced in a wire from the even number $k$ down to the uneven number $k-1$, which in turn allows the wire to treat the clause triangle as its first inner polygon, and we save the three segments of the clause triangle, leading to a cost of $\frac{3(k_a + k_b + k_c)}{2}$ and the variable gadget segment number increases from $12$ to $14$.
    
    Note further that a minimum link fencing would never choose to do so for two or three wires, since the benefit of including the clause triangle in a fence of a wire can only be achieved once.
    Assume that there is a second wire in the satisfying state, which add a further charge of (at least) two additional segments at its variable gadget.
    By Lemma~\ref{lem:only_in_one_wire}, we know that the clause triangle is only included in a fence of one of the two wires.
    The second wire has to fence an odd number of inner polygons and therefore has to include either a fence of size one or three.
    It therefore still requires $\frac{3k}{2}$ segments.
    Therefore only one wire will ever be put in the satisfying state, even if two or all three variables would allow their connected wires to be put in the satisfying state for an additional charge of two segments each.
\end{proof}

\paragraph*{Full correctness proof}
\nphard*
\begin{proof}
    Given an instance of planar $3,4$-\SAT, we create and place a variable gadget $\mathcal{G}(v)$ for every variable $v\in \mathcal V$ and a clause gadget $\mathcal{G}(c)$ for every clause $c$, as described above.
    The wires of $\mathcal{G}(c)$ connect to a true spike of $\mathcal{G}(v)$ if $v$ appears as a positive literal and to a negative spike if $v$ appears as a negative literal in $c$.
    
    Assume we are given a satisfying variable assignment for the instance of planar $3,4$-\SAT.
    By Lemma~\ref{lem:clausecost} in order to be able to find a fencing with $\frac{3(k_a + k_b + k_c)}{2}$ segments for every clause gadget, at least one wire must be in the satisfying state.
    If one wire is in the satisfying state we can include the last inner polygon of the wire for an additional two segments if and only if the corresponding variable gadget is in the proper state (Lemma~\ref{lem:variablecost}).
    Since every clause has one variable, which satisfies this clause, we choose this variable gadget to be in this state, and fence the variable gadget accordingly, leading to twelve segments plus two segments, per clause, which is connected via a wire in the satisfying state to it.
    Since in the variable assignment, every clause has such a literal and every variable is either true or false, we can do this for every clause and are never required to put a variable gadget both in its true and its false state.
    The final cost is $|\mathcal V| \cdot 12 + \sum \limits_{c \in \mathcal C} (\frac{3(k^c)}{2} + 2)$, where $k^c$ is the number of inner polygons of all three wires of $\mathcal{G}(c)$ summed up.

    Now assume that we are given a fencing of the created \textsf{BMLF} instance with exactly $|\mathcal V| \cdot 12 + \sum \limits_{c \in \mathcal C} (\frac{3(k^c)}{2} + 2)$ segments.
    This implies that for every clause gadget, there is one wire gadget in the satisfying state.
    We follow this wire up to the variable gadget, which has to be in the true state if the wire is connected to a true spike and in a false state otherwise.
    We set the corresponding variable of the variable gadget to true in the former and to false in the latter case.
    Therefore every clause has one variable, which satisfies the clause.
    Since no fencing of a variable gadget exists, in which both a wire connected to a true spike and a wire connected to a false spike can be put into the satisfying state (Lemma~\ref{lem:variablecost}), the implied variable assignment is consistent.
    Finally, this might not necessarily assign all variables to a fixed truth assignment, since even if a clause might be satisfiable with two or even all three literals, we will never set more than one wire into the satisfying state.
    All variables, which do not have a truth value assigned yet, can safely be assigned a random value (e.g., true or false if their variable gadget is true or false).
    
    We conclude that the instance of planar $3,4$-\SAT{} is satisfiable if and only if the constructed two-colored \textsf{BMLF} instance admits a fencing with exactly $|\mathcal V| \cdot 12 + \sum \limits_{c \in \mathcal C} (\frac{3(k^c)}{2} + 2)$ segments.

\end{proof}

\section{Complete enumeration of possible fences in a variable gadget}\label{sec:app:variable_gadget}
This section contains the complete enumeration of all possible cases, which are considered in teh proof of Lemma~\ref{lem:var_gadget}.
In particular, we enumerate all possibilities of which $k$ inner polygons could be included in a single fence, but compensate for rotational and axial symmetry, i.e., two groupings are considered rotationally symmetric if we can construct one from the other, by a combination of shifting all indices of included polygons by the same constant (recall that all computations are considered modulo 8, and we write the index 0 and 8 interchangeably) and relabeling all indices $i$ as $8-i$ (mirroring the instance at a straight line).

Clearly there is exactly one possibility for fencing 0, 1, 7 or 8 inner polygons.
Moreover, the number of choosing $k$ polygons to fence is also characterizing how to chose $8-k$ polygons (which are not fenced) and therefore the number of cases is symmetric for 2 and 6, and 3 and 5. It remains to compute the correct number of cases for 2, 3 and 4.

Exclusively accounting for rotational symmetry, we can compute the number $f^r(k)$ with the formula

\[f^r(k) = \frac{1}{8}\cdot \sum \limits_{d|\text{gcd}(k, 8-k)} \phi(d) {{8/d}\choose{k/d}}\]

where $\phi$ is Eulers $\phi$-function, i.e., the number of co-prime integers smaller than $k$ including 1.
Using this formula we obtain four cases for two polygons (non of which are symmetric to each other), seven cases for three polygons (two of which can be eliminated due to the additional axial symmetry) and ten cases for four polygons (three of which are symmetric).%

All cases are shown below in two individual figures, once illustrating the lower bound (the black numbers) on the left, using non-overlapping, non-colinear triples and once showing that this bound is in fact tight, by providing a fence achieving this exact number of segments (shown in blue).

\begin{figure}[h]
    \centering
    \includegraphics[page=9,scale=.57]{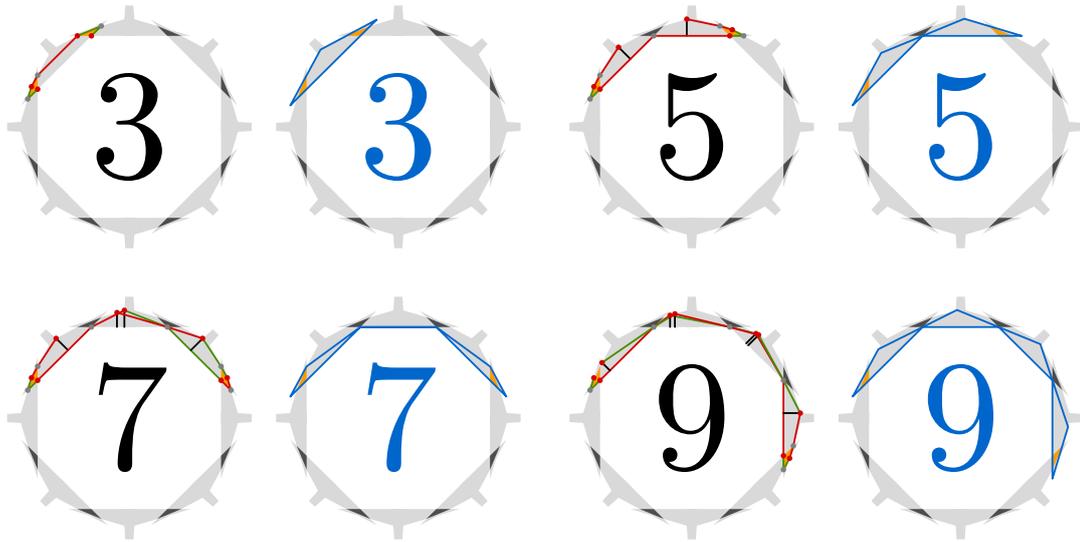}
    \caption{Fences including 2 polygons}
\end{figure}
\begin{figure}[h]
    \centering
    \includegraphics[page=10,scale=.57]{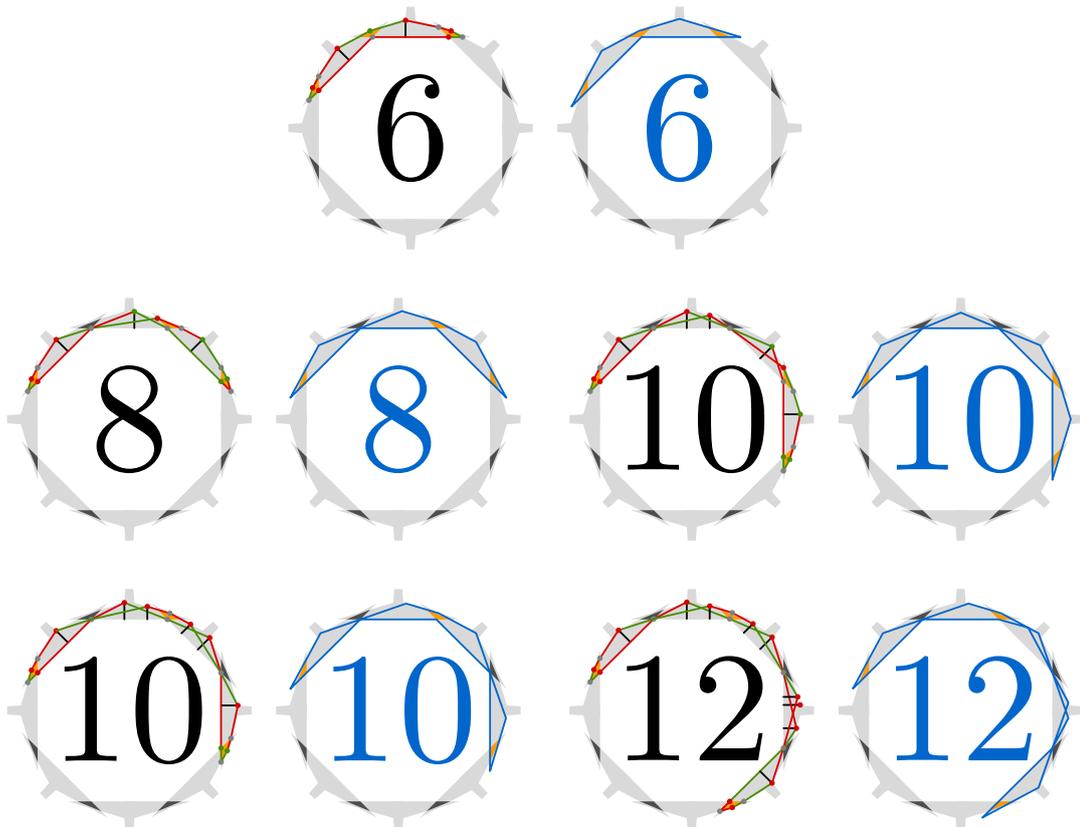}
    \caption{Fences including 3 polygons}
\end{figure}
\begin{figure}[h]
    \centering
    \includegraphics[page=11,scale=.57]{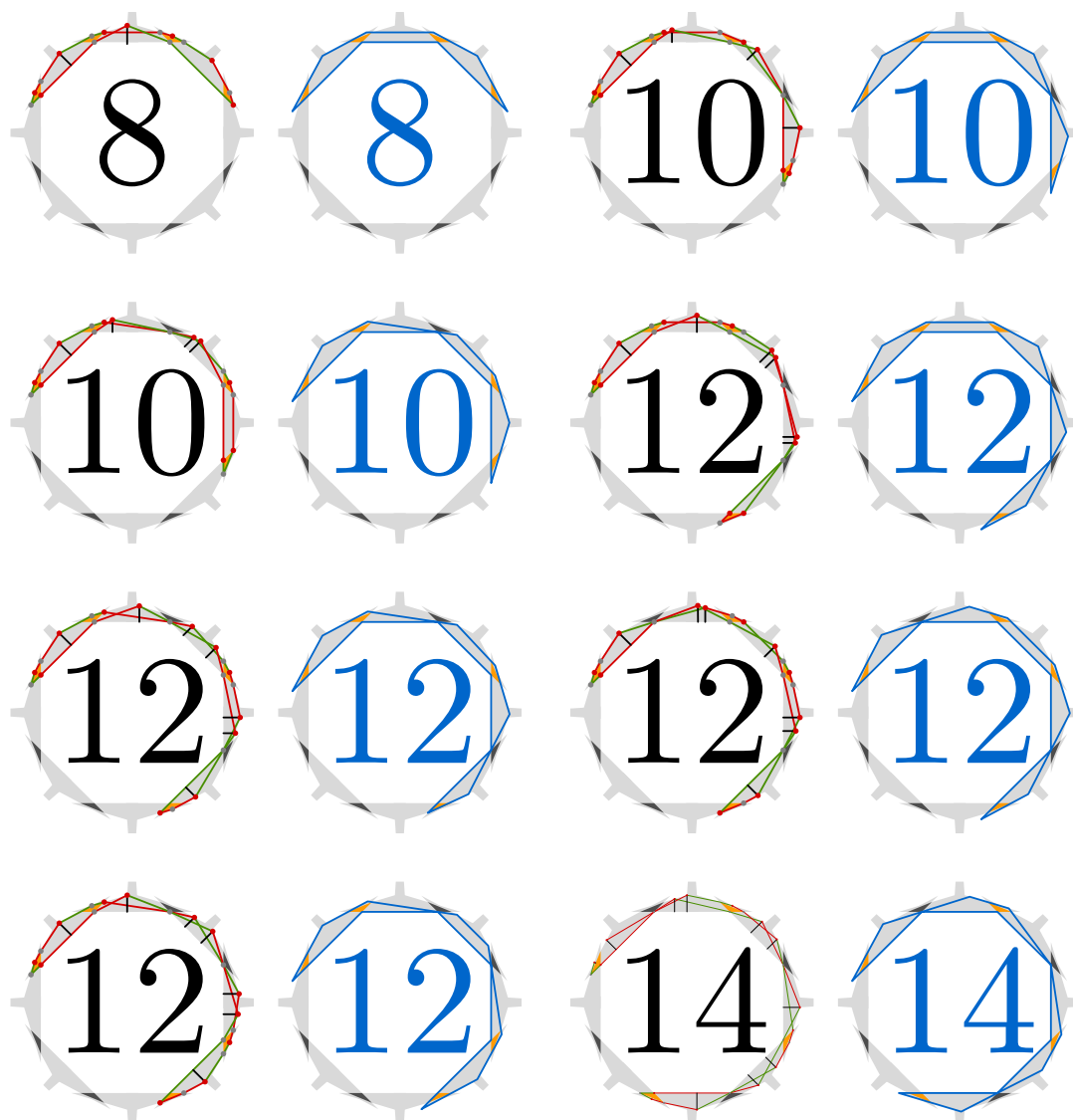}
    \caption{Fences including 4 polygons}
\end{figure}
\begin{figure}[h]
    \centering
    \includegraphics[page=12,scale=.57]{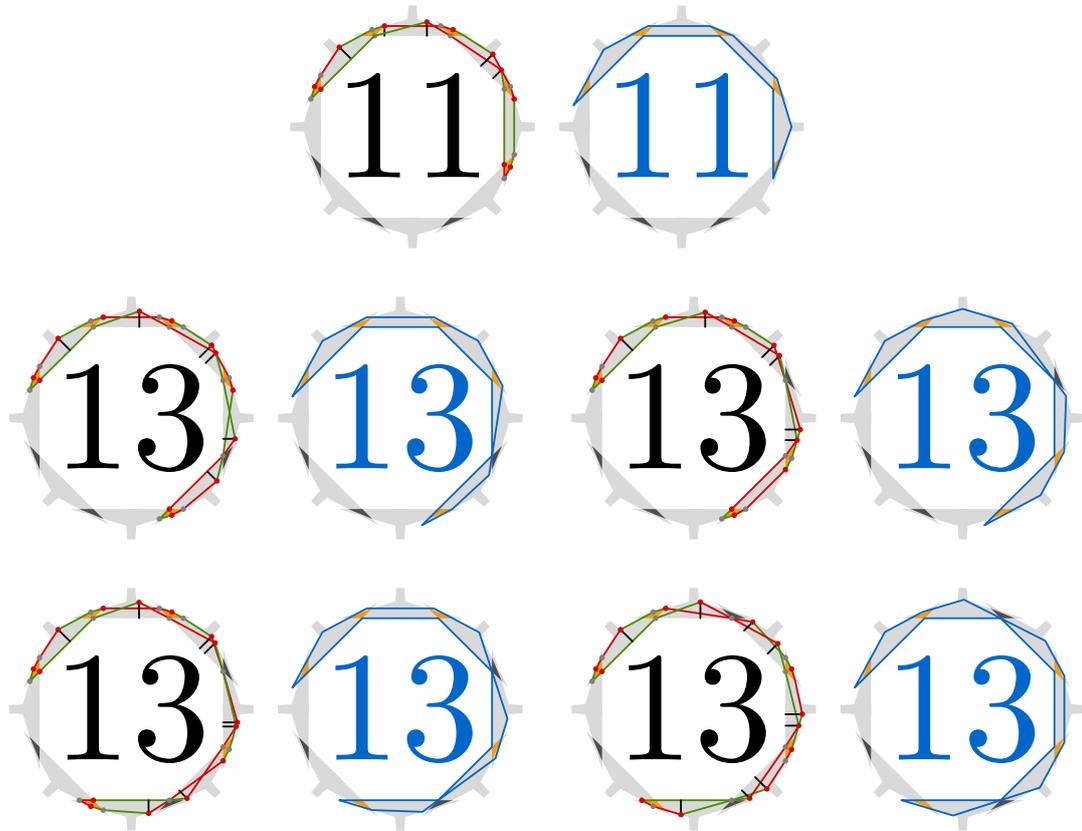}
    \caption{Fences including 5 polygons}
\end{figure}
\begin{figure}[h]
    \centering
    \includegraphics[page=13,scale=.57]{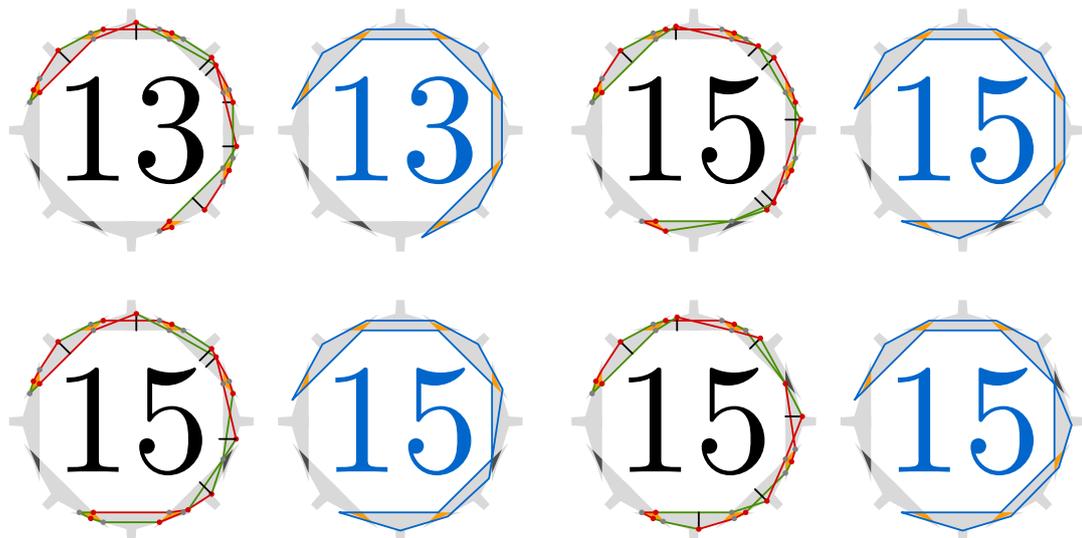}
    \caption{Fences including 6 polygons}
\end{figure}
\begin{figure}[h]
    \centering
    \includegraphics[page=14,scale=.57]{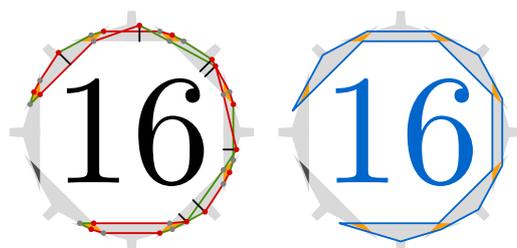}
    \caption{Fences including 7 polygons}
\end{figure}
\begin{figure}[h]
    \centering
    \includegraphics[page=15,scale=.57]{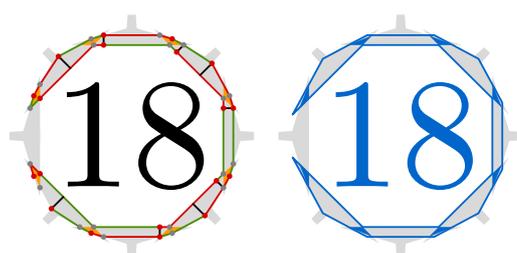}
    \caption{Fences including 8 polygons}
\end{figure}

\clearpage

\section{Complete enumeration of possible fences at the clause triangle}\label{sec:app:clause_triangle}
This section illustrates the enumeration argument of the proof of Lemma~\ref{lem:only_in_one_wire}.
Every case is shown four times in a column.
The first row shows the lower bound on the number of segments for any fence including all polygons of this case, indicated by the three numbers in the upper left corner.
These labels should be read as \textit{XYZ} corresponding to the case $\{B_X^1, B_Y^2, B_Z^3\}$.
The second row shows that all given bounds are in fact tight, as they can be achieved with the shown blue fences.
The third and fourth row show (except for cases 111 and 110), that such a fence can be replaced with a collection fences, which in total achieve a lower number of segments, while either including (3rd row in green) or excluding (4th row in red) the clause triangle, i.e., in both cases, the original fence was not minimal.
\begin{figure}[h]
    \centering
    \includegraphics[page=1,scale=.75]{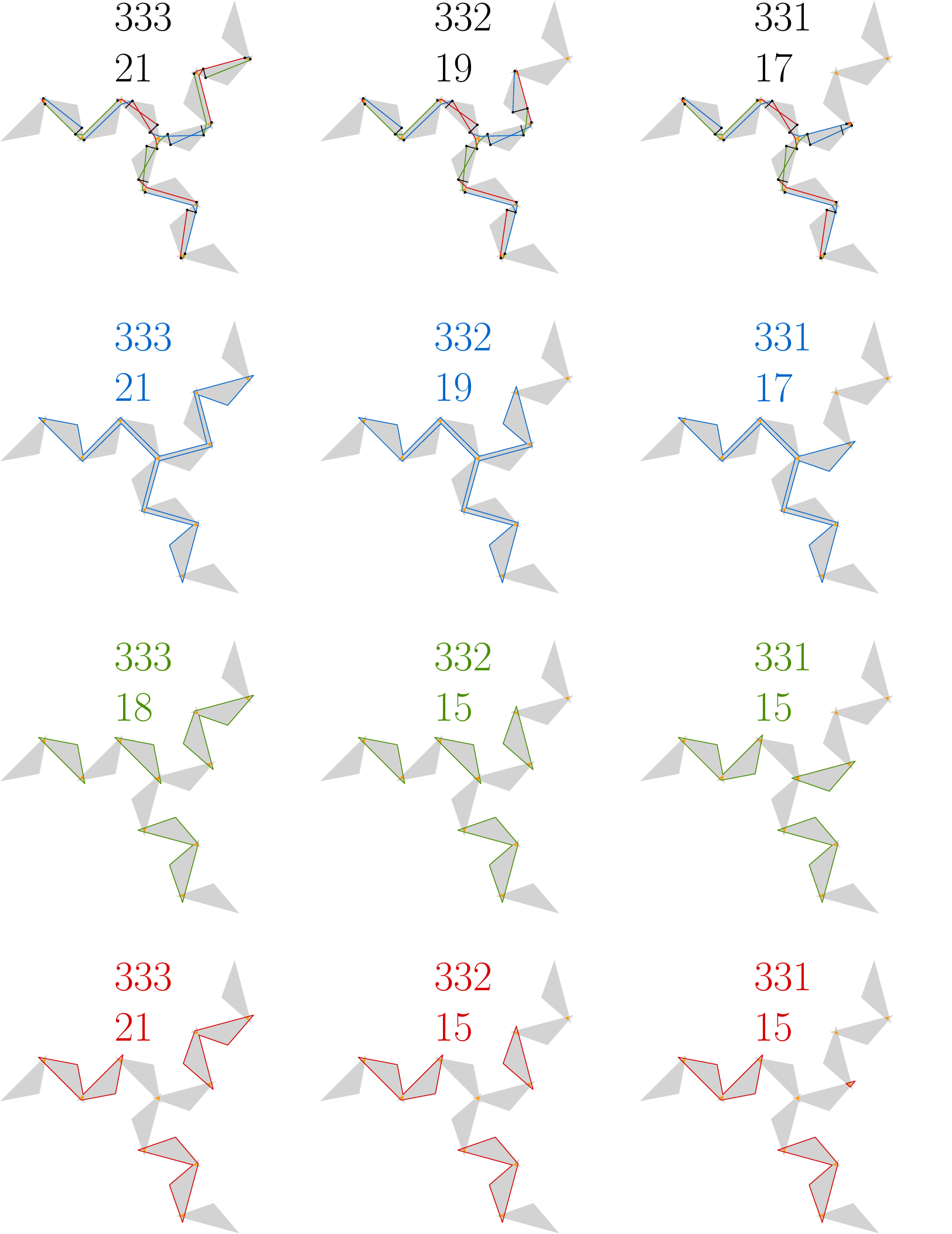}
    \caption{Fences including exactly 3 inner polygons of one and at most 3 of any other wire.}
\end{figure}
\begin{figure}[h]
    \centering
    \includegraphics[page=2,scale=.75]{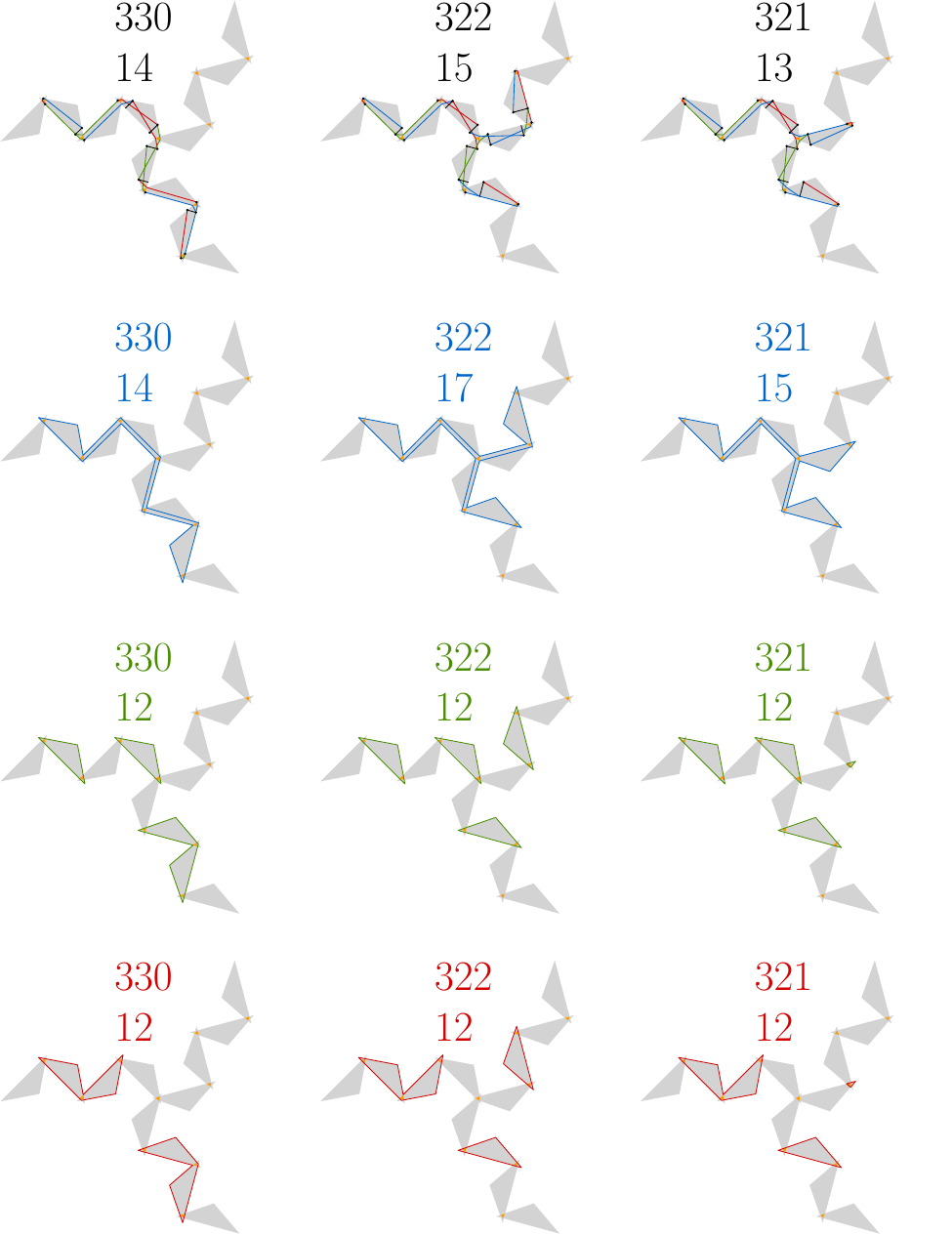}
    \caption{Fences including exactly 3 inner polygons of one and at most 3 of any other wire.}
\end{figure}
\begin{figure}[h]
    \centering
    \includegraphics[page=3,scale=.75]{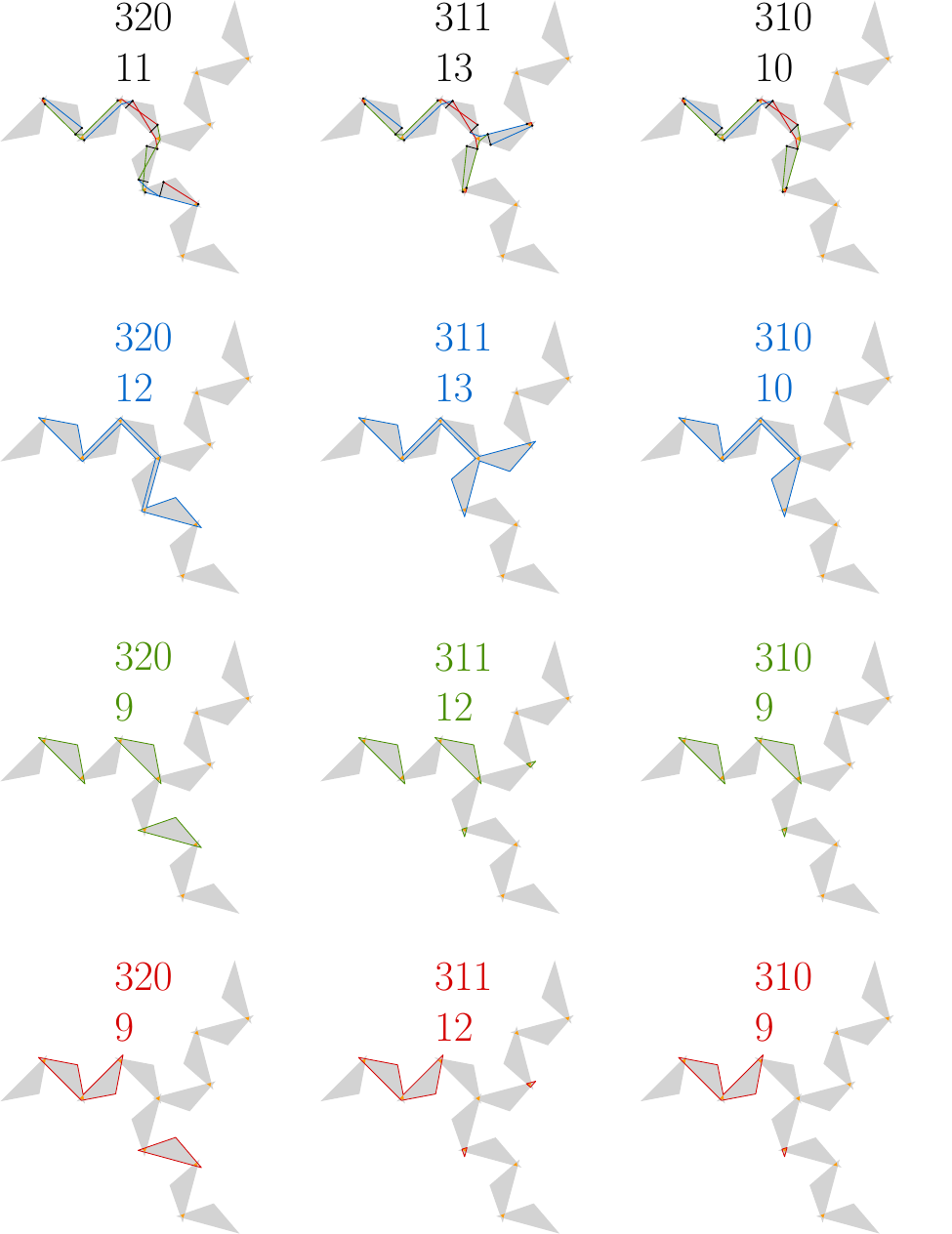}
    \caption{Fences including exactly 3 inner polygons of one and at most 3 of any other wire.}
\end{figure}
\begin{figure}[h]
    \centering
    \includegraphics[page=4,scale=.75]{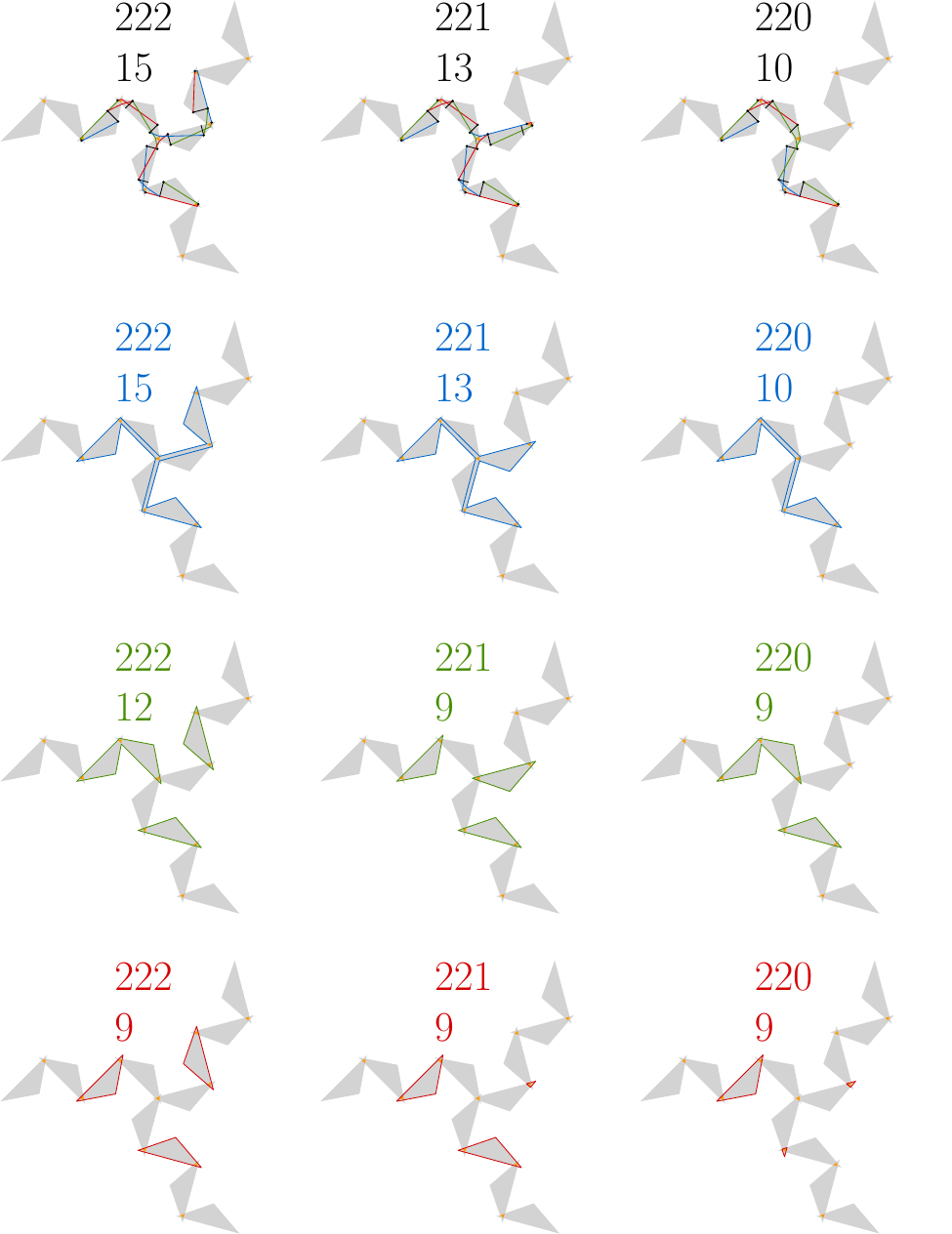}
    \caption{Fences including exactly 2 inner polygons of one and at most 2 of any other wire.}
\end{figure}
\begin{figure}[h]
    \centering
    \includegraphics[page=5,scale=.75]{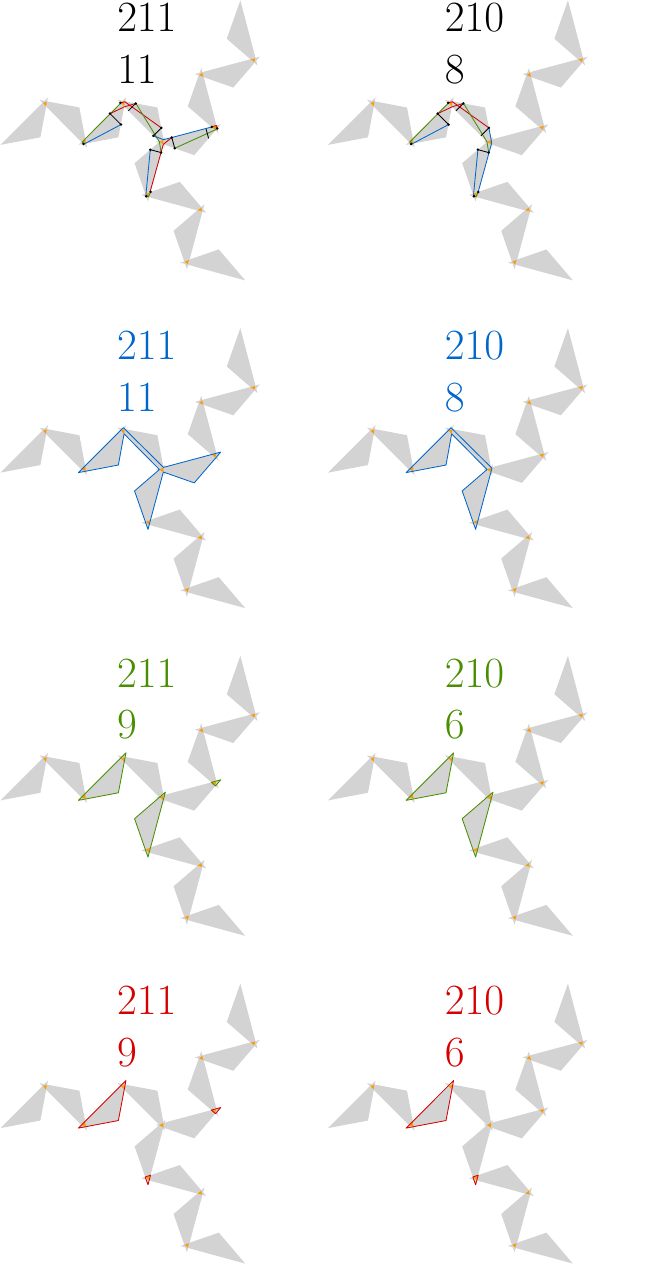}
    \caption{Fences including exactly 2 inner polygons of one and at most 2 of any other wire.}
\end{figure}
\begin{figure}[h]
    \centering
    \includegraphics[page=6,scale=.75]{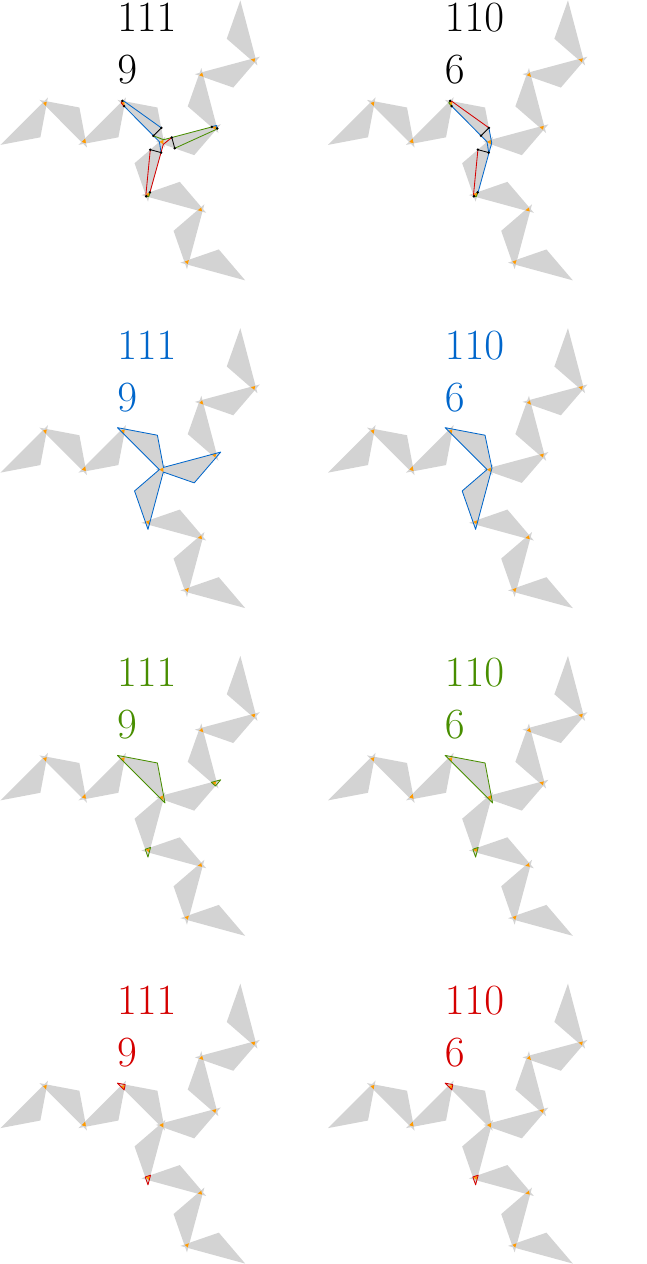}
    \caption{Fences including exactly 1 inner polygons of one and at most 1 of any other wire.}
\end{figure}

\clearpage

\section{Omitted Proofs of Section~\ref{sec:XPalgo}}

\equivalenceclasses*
\begin{proof}
    Let $\mathcal H^\cap \subseteq \mathcal H$ be the set of holes, which intersect the joint convex hull $C$ of $T_1$ and $T_2$ or
    are fully contained in it.
    Let $l$ be a splitting  segment with endpoints $p \in T_1$ and $q \in T_2$.
    Observe, that since $l$ is completely contained in $C - \mathcal H$ we,
    by definition of equivalent splitting segments, can disregard any hole $H \not \in \mathcal H^\cap$.

    Let $l \in C$ be a splitting segment of $T_1$ and $T_2$, 
    which has a set $L \subseteq \mathcal H^\cap$ of holes to its left and
    a set $R \subseteq \mathcal H^\cap$ to its right.
    There exists another splitting segment $l_c$ 
    whose supporting line is a bitangent on two holes, say $H_L \in L$ and $H_R \in R$,
    such that $l_c$ has also $L$ to its left and $R$ to its right.
    We obtain the splitting segment $l_c$ from $l$  by
    rotating $l$ first around its center point
    until it touches either $H_L$ or $H_R$ in a point $r$ and 
    then continue to rotate $l$ around $r$ until it touches the second hole, see \cref{fig:bitangents_a}. 
    Since any pair of holes has at most four bitangents (\cref{fig:bitangents_b}) there are at most 
    $4|\mathcal H^\cap|^2$ equivalence classes.
\end{proof}

\xpalgorithm*
\begin{proof}
    Throughout, we consider the triangulation $\mathcal T$ of the free space of $\mathcal P$.
    Let $t+1$ be the number of colors in the given instance and
    $c_1, \ldots, c_t$ be the colors of polygons in $\mathcal P$ with $c_i \neq f(Q)$ for all $i = 1,\ldots,t$.
    Observe that the homotopy of a fence including exactly one polygon $P \in \mathcal P$ is unique.
    We find a path $\alpha$ with this homotopy in the triangulated free space between the polygons in $\mathcal P$ by traversing the boundary of $P$ clockwise and at every corner of the polygon adding every incident triangle of $\mathcal T$ in clockwise order to a list.
    This yields a series of triangles from which we can construct a loop $\alpha$ 
    such that $P$ is contained in $\alpha$.
    This can be done, for example, by connecting all midpoints of triangulation segments of a triangle and its successor in the loop.
    Then, we use \cref{thm:klinkblackbox} to obtain a minimum-link fence for $P$ from $\alpha$
    or determine that no fence with at most $k$ links exists.
    Computing all individual fences requires $O(|\mathcal{P}|kn) = O(kn^2)$ time.
    
    Next, we consider every pair $P_1,P_2 \in \mathcal P$ with 
    $P_1 \neq P_2$, $c_i = f(P_1) = f(P_2)$ and $i \in \{1,\ldots,t\}$.
    Note that in contrast to a fence containing only one polygon, 
    a fence containing exactly two polygons can belong to several different homotopy-classes w.r.t.
    the remaining polygons. 
    
    Now, we describe how to compute a minimum-link fence $F$ for two polygons $P_1$ and $P_2$ of $\mathcal P$.
    Recall that all polygons in $\mathcal P$ together have $n$ corners.
    Since we have only $O(n)$ triangles in $\mathcal T$
    we can iterate over all possible $O({n \choose k})\in O(n^k)$ ordered tuples of $k$ triangles.
    Fix in the following such an ordered $k$-tuple $(T_1,\ldots,T_k)$ of triangles in $\mathcal T$.
    There are only $O(|\mathcal P|^2)$ many non-equivalent splitting segments 
    connecting points in triangles $T_i$ and $T_{i+1}$ by Lemma~\ref{lem:equivalence_classes}.
    Consequently, we can iterate over the $O((|\mathcal P|^2)^k) = O(n^{2k})$ 
    many different combinations of splitting segments
    between consecutive triangles.
    In case there are two consecutive triangles between which no possible splitting line exists we reject
    this tuple of triangles.
    Assume in the following that we fix for every pair $T_i$ and $T_{i+1}$ a splitting segment $l_i$.

    It remains to construct a plane loop $\alpha$ as input for the algorithm of Hershberger and Snoeyink~\cite{DBLP:journals/comgeo/HershbergerS94} or
    decide that no such loop exists for the fixed choices of triangles and splitting segments.
    From the triangles $T_1,\ldots,T_k$ and the splitting segments $l_1,\ldots,l_k$
    we derive a sequence of triangles $\tau_1,\ldots,\tau_z$ of $\mathcal T$ that $\alpha$ has to visit.
    Since the triangulation $\mathcal T$ is defined by the corners of polygons in $\mathcal P$
    each splitting segment $l_i$ gives rise to a unique sequence of triangles.
    We concatenate all these sequences starting with the sequence induced by $l_1$ 
    to obtain the sequence $\Theta = (\tau_1,\ldots,\tau_z)$.
    Observe, that triangles along this sequence may repeat and that $\tau_1 = \tau_z$.
    
    It remains to decide if there exists a plane loop $\alpha$ visiting each triangle of $\Theta$ in order.
    To make the following description simpler let $s_i$ be the shared boundary of $\tau_i$ and $\tau_{i+1}$.
    If for no $i$ with $i \in \{1,\ldots,z-1\}$, we find that $s_i = s_{i+1}$ we create a loop $\alpha$ by
    connecting the centerpoint of $s_i$ with the one of $s_{i+1}$.
    Since no boundary repeats, this is always possible without any centerpoint and 
    hence triangle being used twice.
    Finally, we add the segment from the centerpoint of $s_z$ to the centerpoint of $s_1$ which 
    is also always possible as $\tau_1 = \tau_z$.
    
    Now assume there exist at least two indices $i$ and $j$ with 
    $i \neq j$ and $i,j\in\{1,\ldots,z\}$ such that $s_i = s_j$.
    Build the loop as before and let $\alpha_1,\ldots,\alpha_z$ be the segments in the constructed loop.
    Since we allow repeated boundaries there exist subsequences 
    among the $\alpha_1,\ldots,\alpha_z$ that are repeated.
    In the following we assume that we only consider inclusion maximal repeated subsequences.
    Let $A = \{\alpha_1',\ldots,\alpha_a'\}$ be one occurrence of such a subsequence of the $\alpha_i$'s
    that is repeated at least once and let $\hat A$ be a different occurrence.
    Let $s_1',\ldots,s_a'$ be the subsequence of the $s_i$'s that correspond to the 
    triangle boundaries passed by the segments in $A$ and $\hat A = \{\hat{\alpha_1},\ldots,\hat{\alpha_a}\}$.
    Now observe that in a plane loop the vertices of $A$ and $\hat A$ have to always appear 
    in the same order along $s_1',\ldots,s_a'$.
    If they would not, let $s_i'$ and $s_{i+1}'$ for $i \in \{1,\ldots,a-1\}$ be two segments
    such that the vertices of $A$ and $\hat A$ on $s_i'$ and $s_{i+1}'$ are not in the same order.
    Without loss of generality assume the vertex of $A$ on $s_i'$ is above the one of $\hat A$ and %
    the opposite is true for $s_{i+1}'$, 
    then we find that $\alpha_i$ and $\hat{\alpha_i}$ cross.
    Hence, the only decision to make is to decide, for each pair of repeated subsequences, in which order 
    their vertices appear along the corresponding triangle-boundaries.
    Since every repeated subsequence implies an intersection 
    between two segments $l_i$ and $l_j$ with $i \neq j$ and $i,j\in \{1,\ldots,k\}$
    we find that there are at most $k^2$ such repeated sequences.
    Consequently, there are at most $2^{O(k)}$ possible ways to distribute
    the center points in each shared part along the boundaries.

    To sum up, for one pair of polygons $P_1,P_2$ we have to consider 
    $O(n^{2k})$ possible non-homotopy equivalent fences and
    for each homotopy we can check in $O(n^2 \cdot 2^k)$ 
    if there exists a plane loop $\alpha$ realizing it, leading to a total runtime of $O(2^k\cdot n^{2k+2})$
    to enumerate every potential homotopy of a minimum link fence.
    For each of the $O(n^{2k})$ different homotopies, we can use \cref{thm:klinkblackbox} to compute a minimum link fence in $O(kn)$, hence we can compute a minimum link fences for all pairs of polygons in $O(n^2\cdot kn^{2k+2}) = O(kn^{2k+4})$.

    If for any polygon no fence, alone or in a pair with another polygon, 
    with $k$ or fewer links is found, 
    we return that no solution exists.
    Otherwise, let $\lambda_{uv}$ be the number of links for a minimum link fence containing $P_u,P_v\in\mathcal P$ and
    $\lambda_u$ the number of links for a minimum link fence containing only $P_u \in\mathcal P$.
    Consider a complete graph $G$ containing one vertex $u$ for each polygon $P_u\in\mathcal P$ and
    one more vertex $x$ if $|\mathcal P|$ is odd.
    Set the edge-weights $w(u,v) = \min\{\lambda_{uv},\lambda_u + \lambda_v\}$ and
    $w(x,u) = \lambda_u$ for $P_u,P_v \in \mathcal P$.
    If for some $P_u \in \mathcal P$ or pair $P_u,P_v \in \mathcal P$ 
    no fence with $\leq k$ segments existed we remove that edge.

    To find a minimum-link fencing of $\mathcal P$ it now suffices to compute a minimum weight perfect matching in $G$.
    Let $M$ be such a matching.
    Then, a minimum link fencing $\mathcal{F}$ of $\mathcal{P}$ can be constructed from $M$ in the following way.
    If $uv \in M$, we add the (pre-computed) minimum link fences containing only $P_u$ and $P_v$ to $\mathcal{F}$ if the weight $w(u,v) = \lambda_u + \lambda_v$ or
    the fence containing $P_u$ and $P_v$ if $w(u,v) = \lambda_{uv}$.
    If $|\mathcal P|$ was odd we also find an edge $xu \in M$ and
    we add the fence containing only $P_u$ to the fencing.
    Finding a minimum weight perfect matching in a general graph with $V$ vertices and $E$ edges can be done for example in $O(V^2E)$ time 
    via finding a maximum weight perfect matching (e.g.~\cite{edmonds_1965}) in the same graph with edge weights set to
    maximum edge-weight plus one minus the original edge-weight.
    Since $G$ has $O(|\mathcal{P}|)$ vertices and $O(|\mathcal{P}|^2)$ edges we can compute this matching in $O(|\mathcal{P}|^4) = O(n^4)$, which is dominated by the initial computation of the minimum link fences.
\end{proof}

\end{document}